\documentclass[journal,final,letterpaper,twoside,twocolumn,romanappendices]{IEEEtran}

\usepackage[sort&compress, numbers]{natbib}
\usepackage{amsmath}
\usepackage{amsthm}
\usepackage{amssymb}
\usepackage{xfrac} 
\usepackage[caption=false,font=footnotesize]{subfig} 
\usepackage{acronym} 
\usepackage{graphicx}
\usepackage{enumerate}
\usepackage{xcolor}
\usepackage{multirow}
\usepackage{hhline} 
\usepackage[para,flushleft]{threeparttable}

\def\TRUE{1}		
\def\CompilationDisplayColor{\TRUE}		

\if\CompilationDisplayColor \TRUE

\else

\fi

\acrodef{SPD}{symmetric positive definite}
\acrodef{HPD}{Hermitian positive definite}
\acrodef{PT}{parallel transport}
\acrodef{AI}{affine-invariant}
\acrodef{LE}{log-Euclidean}
\acrodef{BW}{Bures-Wasserstein}
\acrodef{DoA}{direction-of-arrival}
\acrodef{SoI}{signal-of-interest}
\acrodef{SNR}{signal to noise ratio}
\acrodef{CB}{conventional beamformer}
\acrodef{MVDR}{minimum variance distortionless response}
\acrodef{CF}{covariance fitting}
\acrodef{CRB}{Cram\'er-Rao bound}
\acrodef{RF}{radio-frequency}
\acrodef{ULA}{uniform linear array}
\acrodef{GLC}{generalized sidelobe canceller}
\acrodef{LCMV}{linear constraint minimum variance}
\acrodef{HPBW}{half-power beamwidth}

\def\a{{\bf a}}
\def\n{{\bf n}}

\def\u{{\bf u}}

\def\y{{\bf y}}

\def\I{{\bf I}}
\def\M{{\bf M}}

\def\R{{\bf R}}
\def\U{{\bf U}}

\def\W{{\bf W}}
\def\Rh{{\hat{\bf R}}}

\def\Rb{{\bar{\bf R}}}
\def\0b{{\bf 0}}
\newcommand{\bmat}[1]{\begin{bmatrix}#1\end{bmatrix}} 
\DeclareMathOperator*{\argmin}{argmin}

\DeclareMathOperator*{\Real}{Re}
\DeclareMathOperator*{\Trace}{Tr}
\DeclareMathOperator*{\Det}{Det}
\newtheorem{proposition}{Proposition}

\newtheorem{definition}{Definition}

\begin{document}



\title{
Riemannian Covariance Fitting for Direction-of-Arrival Estimation
}
\author{Joseph S. Picard, Amitay Bar, Ronen Talmon
\thanks{
This work was supported by the European Union’s Horizon 2020 research and
innovation program under grant agreement No. 802735-ERC-DIFFOP.
}
\thanks{The authors are with the Viterbi Faculty of Electrical and Computer Engineering, Technion - Israel Institute of Technology, Haifa 32000, Israel (Corresponding author: Joseph S. Picard, email: josephp@technion.ac.il).}
}


\maketitle


\begin{abstract}

Covariance fitting (CF) is a comprehensive approach for direction of arrival (DoA) estimation, consolidating many common solutions. Standard practice is to use Euclidean criteria for CF, disregarding the intrinsic Hermitian positive-definite (HPD) geometry of the spatial covariance matrices. We assert that this oversight leads to inherent limitations.
In this paper, as a remedy, we present a comprehensive study of the use of various Riemannian metrics of HPD matrices in CF. We focus on the advantages of the Affine-Invariant (AI) and the Log-Euclidean (LE) Riemannian metrics. Consequently, we propose a new practical beamformer based on the LE metric and derive analytically its spatial characteristics, such as the beamwidth and sidelobe attenuation, under noisy conditions. Comparing these features to classical beamformers shows significant advantage. In addition, we demonstrate, both theoretically and experimentally, the LE beamformer's robustness in scenarios with small sample sizes and in the presence of noise, interference, and multipath channels.

\end{abstract}



\section{Introduction}
\label{Section:Introduction}

\IEEEPARstart{T}{he} problem of \ac{DoA} estimation has gained significant attention from the research community in recent decades, thanks to the emergence of new challenging applications. 
In the context of radio signals, 
the \ac{DoA} problem consists of estimating the direction of emitters using their signals intercepted by a phased-array built of synchronized sensor elements placed at known locations. 
\ac{DoA} estimation is long-standing and well researched \cite{VanTrees,stoica2005spectral,BibList:Balanis,BibList:Mailloux,BibList:Wax,krim1996two,BibList:Pesavento,chakrabarty2017broadband,chakrabarty2019multi}, yet, it is still considered an open problem with many outstanding challenges involving high levels of noise, small sample size, multiple sources, and multipath channels.
Classical \ac{DoA} estimators were originally developed by invoking statistical, spectral, or spatial principles to design beamformers with some desired properties \cite{krim1996two,BibList:Li,BibList:Schmidt,BibList:WeissMODE,BibList:Viberg,stoica2005spectral}. 
Subsequently, it has been shown that many of these classical algorithms can be recast as \ac{CF} problems \cite{BibList:Ottersten0}, or as variants of \ac{CF}, such as signal or noise subspace fitting problems \cite{BibList:Ottersten,BibList:WaxSSM}. 
Since then, \ac{CF} has assumed a central role in \ac{DoA} estimation both directly \cite{BibList:Pesavento,BibList:Ottersten0,BibList:Trinh,BibList:Schenck,BibList:Yardibi} and indirectly \cite{BibList:Viberg,BibList:Ottersten,BibList:WaxSSM,BibList:Picard,BibList:Zheng,BibList:Delmer}. 
For example, recently, \ac{CF} was used as the underlying principle for improved or sparse representation of spatial spectrum \cite{BibList:Trinh,BibList:Schenck,BibList:Yardibi}, both with LASSO regularization \cite{BibList:Lasso} or without \cite{BibList:Picard,BibList:Zheng,BibList:Delmer}.

\IEEEpubidadjcol 

\ac{CF} estimates the \ac{DoA} by identifying the direction at which the sample covariance matrix of the intercepted signal best matches a direction-dependent model. Mathematically, it consists of fitting the sample covariance matrix $\Rh$ to its direction-dependent model $\R(\theta)$, where $\theta$ is the direction\footnote{For brevity, we will omit the explicit dependence of $\R$ on $\theta$ in the sequel.}.
For example, the \ac{CB} (or delay-and-sum beamformer) for \ac{DoA} estimation, which is characterized by robustness to small sample sizes, can be recast using \ac{CF} by minimizing the Frobenius norm of the difference $\|\Rh-\R\|{}_F$, or from a geometric standpoint, their Euclidean distance. Similarly, the \ac{MVDR} beamformer, which is also known as the Capon beamformer and is typically characterized by a narrow mainlobe, is obtained by minimizing $\|\Rh^{-1}-\R^{-1}\|_F$ (see details in Section \ref{Section:ProblemFormulation}). 
The commonality between the two minimization criteria raises the question whether the Euclidean metric is necessarily the most appropriate metric for \ac{CF}. 
Since the Euclidean distance (i.e., the Frobenius norm) between matrices equals the Euclidean distance between their vectorization and since it is computed element-wise, we argue that alternative distances that specifically take into account the structure of covariance matrices should be considered. Subsequently, we foster the search and analysis of such distances in the context of \ac{CF}.

Another argument that illustrates the potential inadequacy of the Euclidean distance is the following. 
The HPD covariance matrices $\Rh$ and $\R$ are invertible, and since matrix inversion is a lossless manipulation that preserves the underlying information of full-rank matrices, then $\Rh$ and $\R$ embody the same information as their inverse matrices.
Nevertheless, the two criteria, $\min \|\Rh-\R\|{}_F$ and $\min \|\Rh^{-1}\!-\!\R^{-1}\|_F$, yield two fundamentally different solutions (the \ac{CB} and the \ac{MVDR} beamformer, respectively, with an inherent tradeoff between mainlobe width and robustness to small sample sizes), while, from an information standpoint, they are similar.
Furthermore, many standard signal processing procedures apply an invertible weighting matrix $\W$ to the array-received signals, for example, tapering (windowing functions, e.g., Hamming, Blackman, etc.) that lowers the sidelobe levels at the expense of mainlobe broadening \cite{VanTrees}.
Indeed, $\|\Rh-\R\|{}_F \neq \| \W\Rh\W^H \!-\! \W\R\W^H\|_F$, and therefore, through \ac{CF}, lead to different beamformers with an embedded tradeoff between the sidelobe level and the mainlobe width.
It is therefore plausible that the two tradeoffs indicated above stem from the use of the Euclidean distance. We will show that using non-Euclidean distances between the covariance matrices mitigates these tradeoffs.

The search for an alternative distance for \ac{CF} naturally aligns with the particular geometry of covariance matrices. Indeed, the covariance matrices $\Rh$ and $\R$ are \ac{HPD}, thereby lying on the Riemannian manifold constituted by the \ac{HPD} space 
endowed with a Riemannian metric \cite{pennec2006riemannian,BibList:Bhatia}.
%
%
Research involving the geometry of \ac{HPD} matrices has drawn a lot of interest in recent years, introducing a plethora of new Riemannian geometries that include new metrics, geodesics, and distances, defined on the space of \ac{HPD} matrices \cite{thanwerdas2022bures,han2021generalized}. These Riemannian geometries have been employed in a broad range of applications, e.g., in computer vision and medical data analysis, showing significant advantages over their Euclidean counterparts \cite{barachant2011multiclass,barachant2013classification,rodrigues2017dimensionality,rodrigues2018riemannian,yair2019parallel}.

In this paper, we present a general framework for \ac{CF}, which accommodates both Euclidean and Riemannian metrics. 
We utilize this framework to specify the advantages of the \ac{AI} metric \cite{BibList:Bhatia,pennec2006riemannian,thanwerdas2019affine}, which is arguably the most popular Riemannian metric, over the Euclidean metric in the context of \ac{DoA} estimation problems. In addition, we study other Riemannian metrics and establish their exact relationships to classical \ac{DoA} methods. Then, we focus on the \ac{LE} metric \cite{BibList:Arsigny,arsigny2007geometric}, which approximates the \ac{AI} metric on the Riemannian manifold of \ac{HPD} matrices and has more efficient implementations, and propose a new beamformer induced by this \ac{LE} metric. 
We analytically derive closed-form expressions for the spatial characteristics of the new \ac{LE} beamformer, such as the beamwidth and sidelobe attenuation. We then compare these expressions with that of classical estimators in noisy conditions using new analytic expressions that account for noise, thereby revealing its advantages over classical estimators in terms of these characteristics.
We theoretically and experimentally showcase the robustness of the \ac{LE} beamformer in scenarios characterized by small sample size, the presence of noise and interfering sources, and multipath channels.
Expanding our investigation beyond Riemannian metrics, we also examine the utilization of other non-Euclidean metrics, such as the Kullback-Liebler divergence \cite{BibList:Chevallier}, 
and the Log-Determinant distance \cite{sra2012new}.

It is noteworthy that in the last few years, several pioneering papers on \ac{CF} using Riemannian distances rather than the classical Euclidean distance have been published. 
However, a comprehensive framework for the Riemannian approach has not yet been presented. 
Indeed, in \cite{BibList:Coutino}, Coutino et al. introduced a new beamformer, which was inspired by Information Geometry and is tightly related to the \ac{AI} metric. However, this beamformer considers only a degenerate approximation of the \ac{AI} distance, and it also requires incident signals with unit power. In \cite{BibList:Dong}, the authors partly circumvented this overly stringent requirement by introducing a scaling parameter. In their work, the sources are not required to have unit power, but they are still required to have identical power. In \cite{BibList:Chahrour1}, Chahrour et al. introduced a \ac{DoA} estimator for radar signals relying on the definition of the AI distance instead of its \ac{LE} approximation. 
However, this method holds only in the case of Toeplitz covariance matrices, and the Toeplitz property is exploited in a preliminary step for the evaluation of the reflection coefficients of incident radar signals. 
As a consequence, the incident power is treated as a known parameter during the \ac{DoA} estimation procedure. 
Another related work is \cite{BibList:AmirWeiss}, where the Kullback-Leibler divergence between distributions was considered as a criterion for \ac{CF} in the context of \ac{DoA} estimation of acoustic emitters. This has a tight connection to the present work because the \ac{AI} distance between covariance matrices can be cast as a Fisher-Rao divergence between normal distributions parameterized by these matrices \cite{BibList:Pinele}.
Extending the scope of \ac{CF}, the Riemannian mean of covariance matrices has also been considered in \cite{bar2023interference}, where improved robustness against interfering sources in \ac{DoA} estimation was shown. 


With respect to these existing works \cite{BibList:Coutino,BibList:Dong,BibList:Chahrour1}, our contribution is threefold.
First, we present a geometric \ac{CF} framework for \ac{DoA} estimation, which provides new insights into existing methods and facilitates the development of a new class of Riemannian geometry-based methods. We exploit this framework to highlight the advantages of Riemannian metrics over the Euclidean metric.
Second, we propose a new \ac{DoA} estimator that relies on the \ac{LE} Riemannian metric and compare it analytically to other metrics and beamformers using closed-form expressions, thereby demonstrating analytically its advantages.
Third, we show that the spatial characteristics (e.g., the beamwidth and sidelobes attenuation) of the proposed \ac{LE} estimator are advantageous at low \ac{SNR} and small sample size.

The remainder of this paper is structured as follows. 
In Section \ref{Section:ProblemFormulation}, we recast classical beamformers as \ac{CF} problems involving the Euclidean metric, 
thereby defining a \ac{CF} framework for \ac{DoA} estimation. 
We highlight some limitations that stem from the Euclidean metric when used as a fitting criterion in \ac{DoA} estimation problems, 
and we present how Riemannian metrics overcome these limitations, putting a special focus on the \ac{AI} metric. 
In Section \ref{Section:LogEuclideanBeamformer}, 
we present the \ac{LE} beamformer, 
derived for \ac{DoA} estimation, along with analytical expressions for some characteristics of the spatial spectrum. 
In Section \ref{Section:Alternative_Riemannian_Metrics}, 
we show that classical non-Euclidean distances of \ac{HPD} matrices systematically yield the \ac{CB} and the \ac{MVDR} beamformer, or variants of these two, thereby further demonstrating the novelty of the \ac{LE} beamformer.
In Section \ref{Section:Simulations_and_Numerical_Examples}, we present simulations that showcase the performance and analysis of the proposed \ac{LE} beamformer.

\section{Covariance Fitting Framework for DoA Estimation}
\label{Section:ProblemFormulation}

\subsection{Notations \& Acronyms}\label{Section:Notations}
Unless stated otherwise, we use the following notation. Lowercase letters are used for scalars (e.g., $s$), uppercase letters for integers (e.g., $K$), lowercase bold letters for vectors (e.g., ${\bf v}$), and uppercase bold letters for matrices (e.g., ${\bf M}$).
The conjugate, transpose, conjugate transpose, and pseudo-inverse operators are denoted by the superscripts $(\cdot)^*$, $(\cdot)^T$, $(\cdot)^H$, and $(\cdot)^\dagger$, respectively.

\subsection{Data Model \& Classical Estimators}
Consider a phased-array consisting of $M$ sensors, and let $\a_\theta$ denote the array's unit-norm response vector to a wave propagating from direction $\theta$. Consider also a far-field source transmitting a narrow-band signal from direction $\theta_1$. The signal collected by the phased-array at discrete time $k\in\{1,\dots,K\}$ is the complex vector $\mathbb{C}^M \ni \y_k = \a_1 x_k + \n_k$, 
where $\a_1\triangleq\a_{\theta_1}$ is the phased-array response to the incident signal propagating from direction $\theta_1$, the complex scalar $x_k$ is the $k$-th sample of the unknown incident signal, the vector $\n_k\in\mathbb{C}^{M}$ is a zero-mean Gaussian measurement noise with known covariance $\sigma_n^2 \I$, and $K$ is the number of time samples. Without loss of generality, suppose the noise level is $\sigma_n^2=1$. 
The \textit{sample} covariance matrix of the phased-array received signal is given by
\begin{align}
 \Rh\triangleq\frac{1}{K}\sum_{k=1}^K \y_k\y_k^H,   
 \label{Eq:CovSample}
\end{align}
and it commonly serves as the input of DoA estimation algorithms.
The \textit{population} covariance matrix, obtained by considering the expected value of the sample covariance matrix, is of great interest for theoretical analyses since it has an explicit expression, given by
\begin{align}
\Rb\triangleq\mathbb{E}\left[\Rh\right]=\sigma_1^2 \a_1\a_1^H + \I,
\label{Eq:CovPopulation_without_interference}
\end{align}
where $\sigma_1^2$ is the incident power of the source. 
The \ac{SNR} is equal to $\sigma_1^2$ since the noise level is $\sigma_n^2=1$. 
The population covariance matrix adheres to the following model of a signal propagating from direction $\theta$ with incident power $\sigma^2$,
\begin{align}
\R&=\sigma^2 \a_\theta\a_\theta^H + \I. 
\label{Eq:CovModel}
\end{align}
The dependence of $\R$ on $\theta$ and $\sigma^2$ is omitted for brevity. 

Two classical \ac{DoA} estimators are the \ac{CB}, also referred to as Bartlett or the Delay-and-Sum beamformer \cite{krim1996two}, and the \ac{MVDR} beamformer, referred to as Capon beamformer \cite{capon1969high}. 
These beamformers consist of evaluating the spatial spectra $P_\text{CB}(\theta)$ and $P_\text{MV}(\theta)$ defined in their standard form by \cite{stoica2005spectral}
\begin{align}
    P_\text{CB}(\theta) &= \a_\theta^H \:\Rh\: \a_\theta,  \label{Spectrum:CB} \\
    P_\text{MV}(\theta) &= \displaystyle \frac{1}{\a_\theta^H \Rh{}^{-1} \a_\theta}, \label{Spectrum:MV}
\end{align}
where MV stands for MVDR for conciseness.

\subsection{Direction-Finding by Covariance Fitting}
\label{Section:DirectionFinding_by_CovarianceFitting}

The \ac{CB} and \ac{MVDR} beamformer have been originally derived following statistical, spatial, and spectral considerations. A different and, arguably, more intuitive approach is \ac{CF}. Following algebraic considerations, it simply consists of matching the covariance matrix $\Rh$ to its theoretical model $\R$ given in (\ref{Eq:CovModel}).
The goodness-of-fit criterion for two matrices $\M_1$ and $\M_2$ is usually the Frobenius norm of their element-wise difference $\M_1-\M_2$, i.e., the Euclidean distance of their vectorizations,
\begin{align}
    d_\text{E}\left( \M_1,\M_2 \right)\triangleq\left\| \M_1 - \M_2 \right\|_F.
    \label{Eq:Define_Euclidean}
\end{align}
We show in Appendix \ref{Appendix:CovFit_with_noise} that the \ac{CB} and \ac{MVDR} beamformer consist of fitting $\Rh$ to $\R$, and $\Rh{}^{-1}$ to $\R{}^{-1}$, respectively, 
\begin{align}
    P_\text{CB}(\theta) &=\displaystyle \argmin_{\sigma^2} \:\:d_\text{E}(\Rh,\R),     \label{Eq:CF_CB} \\
    P_\text{MV}(\theta) &=\displaystyle \argmin_{\sigma^2} \:d_\text{E}(\Rh{}^{-1},\R{}^{-1}). \label{Eq:CF_MV}
\end{align}
In other words, the \ac{CB} in (\ref{Spectrum:CB}) and \ac{MVDR} beamformer in (\ref{Spectrum:MV}) can be recast as \ac{CF} problems in (\ref{Eq:CF_CB}) and (\ref{Eq:CF_MV}), respectively, when the fitting criterion is the Euclidean distance.


On the one hand, using the Frobenius norm of the difference as a goodness-of-fit criterion is efficient and practical. On the other hand, this criterion operates element-wise and does not take into account that the correlation matrices are not any matrices but \ac{HPD}, thereby neglecting important information. 


\subsection{Covariance Fitting based on the Affine-Invariant Metric}
\label{Section:Invoking_Riemannian_Metrics}

Here, we focus on the \ac{AI} metric \cite{pennec2006riemannian,BibList:Bhatia}. 
Formally, the distance between the \ac{HPD} covariance matrices $\Rh$ and $\R$ induced by the \ac{AI} metric is given by \cite{pennec2006riemannian,bhatia2006riemannian}
\begin{align}
d_\text{AI}^2\left( \Rh,\R \right) &= 
\left\| \log(\Rh{}^{-\sfrac{1}{2}}\R\Rh{}^{-1/2}) \right\|_F^2, \label{Eq:Define_Metric_AI} \\
&= \sum_{m=1}^M \log^2\left(\lambda_k\left\{\Rh{}^{-\sfrac{1}{2}}\R\Rh{}^{-\sfrac{1}{2}}\right\}\right), \label{Eq:AI_model_using_Eigenvalue_sum}
\end{align}
where $\lambda_m\{\cdot\}$ is the $m$-the largest eigenvalue of a matrix.

Then, we use the \ac
{AI} distance instead of the Euclidean distance in (\ref{Eq:CF_CB}) and (\ref{Eq:CF_MV}), and the spatial spectrum obtained by minimizing the \ac{AI} distance is defined by
\begin{align}
P_\text{AI}(\theta) 
&= \argmin_{\sigma^2}\!\!\! \sum_{m=1}^M \!\log^2\left( \lambda_m\left\{\Rh{}^{-\sfrac{1}{2}}\R\Rh{}^{-\sfrac{1}{2}}\right\}\right).
\label{Eq:Estimator_AI}
\end{align}
We assert that such a distance 
is likely a better goodness-of-fit criterion for \ac{CF} purposes in \ac{DoA} estimation problems than the Euclidean distance for two reasons.
First, the \ac{AI} distance exploits the \ac{HPD} structure of covariance matrices. 
Second, more specifically, minimizing each Euclidean criterion $d_\text{E}( \Rh\!,\R\!)$, $d_\text{E}( \Rh{}^{-1}\!\!\!,\R{}^{-1}\! )$, or $d_\text{E}(\W\Rh\W^H\!,\W\R\W^H\!)$, where $\W$ is an invertible matrix, leads to a different solution, thereby giving rise to the abovementioned tradeoffs between the sidelobe level, the mainlobe width, and the robustness to small samples. Conversely, the invariance of the \ac{AI} to inversion and congruence \cite{BibList:Chevallier} implies that 
$d_\text{AI}( \Rh\!,\R\!)\!=\!d_\text{AI}( \Rh{}^{-1}\!\!\!,\R{}^{-1}\! )\!=\!d_\text{AI}(\W\Rh\W^H\!,\W\R\W^H\!)$, consolidating the solutions and mitigating any tradeoff between them.
Indeed, we show empirically that using the \ac{AI} distance as a fitting criterion leads to a beamformer with spatial spectrum with a narrow mainlobe, low sidelobes, and insensitivity to small sample sizes (see Fig. \ref{Fig:1} in the sequel).

Importantly, we note that the problem in (\ref{Eq:Estimator_AI}) presents a significant advantage over previous work because no assumption is made on $\sigma^2$, and it is considered as an unknown parameter of the model covariance $\R$.  
Conversely, in \cite{BibList:Coutino,BibList:Dong,BibList:Chahrour1}, it was assumed that all the signals have unit power or identical power, or even that the incident power is estimated in a preliminary procedure \cite[RBA Algo.]{BibList:Chahrour1}. Furthermore, the problem in (\ref{Eq:Estimator_AI}) embodies another possible advantage compared with \cite{BibList:Coutino,BibList:Dong}, since it exploits all the eigenvalues instead of using merely the largest one. 
Finally, in contrast to the methods proposed in \cite{BibList:Coutino,BibList:Dong,BibList:Chahrour1}, the spatial spectrum in (\ref{Eq:Estimator_AI}) does not suffer from scaling issues (see more details in Section \ref{Section:Earlier_Riemannian_Approches}).


\section{The Log-Euclidean Beamformer}
\label{Section:LogEuclideanBeamformer}


Given the data model $\R=\sigma^2\a_\theta\a_\theta^H+\I$ in (\ref{Eq:CovModel}), our assumption is that the problem in (\ref{Eq:Estimator_AI}) has no closed-form solution. Therefore, evaluating $P_\text{AI}(\theta)$ requires an exhaustive search over a wide range of values of $\sigma^2$. Performing this search for each $\theta$ requires immense computational resources, especially because the eigenvalues have to be re-evaluated for each power hypothesis $\sigma^2$. 
As a remedy, we propose to consider the \ac{LE} metric \cite{BibList:Arsigny}, which approximates the \ac{AI} metric, and to derive a new beamformer that relies on it.
In this section, we present the derivation and analysis of the beamformer based on the \ac{LE} metric. 
The \ac{LE} distance between the \ac{HPD} covariance matrices $\Rh$ and $\R$ is defined by
\begin{align}
d_\text{LE}^2\left( \Rh,\R \right) &= 
\left\| \log(\Rh) - \log(\R) \right\|_F^2.
\label{Eq:Define_Metric_LE}
\end{align}
For simplicity, we use $\log$ both as scalar logarithm and logarithm of a matrix, depending on the context.
Similarly to the \ac{AI} metric, the \ac{LE} metric is invariant to inversion \cite{BibList:Arsigny}.  
Additionally, as a tight approximation of the \ac{AI} metric, the \ac{LE} metric also approximately admits the invariance to congruence, but not \textit{stricto sensu} \cite{BibList:Arsigny}. 
\begin{definition}[Log-Euclidean Beamformer]
\label{Definition_LE}
The \ac{LE} beamformer is defined as the minimizer of the \ac{LE} distance between the covariance matrix $\Rh$ and the model $\R$, i.e.,
\begin{align}
    P_\textnormal{LE}(\theta) &= \argmin_{\sigma^2} d_\textnormal{LE}^2\left( \Rh, \R \right).
    \label{Eq:Proposition_LE}
\end{align}
\end{definition}

\begin{proposition}[Explicit expression of the LE beamformer]
\label{Proposition_LE}
If the model is $\R=\sigma^2\a_\theta\a_\theta^H+\I$ as in (\ref{Eq:CovModel}), 
then the \ac{LE} spatial spectrum has a closed-form expression given by
\begin{align}
    P_\textnormal{LE}(\theta)=\exp\left(\a_\theta^H \log(\Rh) \a_\theta\right)-\!1.
    \label{LE_PowerSpectrum}
\end{align}
\end{proposition}
\noindent The proof appears in Appendix \ref{Appendix:LE_Spatial_Spectrum}.


In the remainder of this section, we present the spatial characteristics of the \ac{LE} beamformer and provide closed-form expressions for the fadings and sidelobes directions, the \ac{HPBW}, the peak-to-sidelobe ratio (PSLR), and the sensitivity to multipath. In addition, we demonstrate the advantages of the LE spectrum over the CB and MVDR spectra in terms of these characteristics. 
To this end, for analysis purposes, we consider the population covariance matrix ${\Rb=\sigma_1^2 \a_1\a_1^H + \I}$ defined in (\ref{Eq:CovPopulation_without_interference}) instead of the sample covariance matrix $\Rh$ defined in (\ref{Eq:CovSample}). In this case, the spatial spectrum $P_\text{LE}(\theta)$ reduces to (see Appendix \ref{Appendix:LE_Fadings_Sidelobes})
\begin{align}
    P_\text{LE}(\theta) 
    & = (\sigma_1^2+1)^{b_\theta}-1,\label{Analysis_LE_SpatialSpectrum_Final}
\end{align}
where $b_\theta$ denotes the beampattern \cite{VanTrees} or array factor \cite{BibList:Balanis}, defined as the response of the phased-array in direction $\theta$ to an incident wave propagating from an arbitrary direction $\theta_1$, 
\begin{align}b_\theta\triangleq\left|\a_\theta^H\a_1\right|^2.
    \label{Generic_b}
\end{align}
Equation (\ref{Analysis_LE_SpatialSpectrum_Final}) indicates that the spectrum $P_\text{LE}(\theta)$ depends on $\theta$ through $b_\theta$ exclusively. Similarly, one can check that this observation also applies to the spectra $P_\text{MV}(\theta)$ and $P_\text{CB}(\theta)$ evaluated using the population covariance matrix, thereby revealing that the variable of interest is actually $b_\theta$.
Therefore, we propose to characterize and compare the LE, MVDR and CB beamformers through the lens of the beampattern $b_\theta$ instead of $\theta$. We will demonstrate that considering $b_\theta$ also offers the advantage of yielding SNR-dependent characteristics (further details by the end of this section), which contribute to highlight the advantages of the LE beamformer at low SNR.

\begin{proposition}[Fadings and sidelobes of the LE beamformer]
\label{Definition_Fading}
Since the LE spectrum in (\ref{Analysis_LE_SpatialSpectrum_Final}) monotonically increases with $b_\theta$, then the fadings and the sidelobes of $P_\textnormal{LE}(\theta)$ occur at the same directions as those of the beampattern $b_\theta$. 
\end{proposition}
As an example, for a \ac{ULA} with $M$ elements separated by $d_\lambda$ wavelengths, the beampattern takes the following form of a Dirichlet-like kernel \cite{stoica2005spectral},
\begin{align}
    b_\theta = \frac{\;\;\;\sin^2\left(M\pi d_\lambda\;(\cos{\theta}-\cos{\theta_1})\right)}
    {M^2\sin^2\left(\;\;\pi d_\lambda\;(\cos{\theta}-\cos{\theta_1})\right)}.
    \label{Beampattern_Definition}
\end{align}
The fading and sidelobe directions of the beampattern $b_\theta$ 
are described by closed-form expressions obtained by identifying its local minima and maxima \cite{stoica2005spectral}. Based on these expressions, it follows from Proposition \ref{Definition_Fading} that the fading and sidelobe directions of the LE spectrum are
\begin{align}
    \theta\approx\arccos\left(\cos{\theta_1}\pm (\ell\!+\!\zeta)\!/\!(\!M d_\lambda\!)\right)
    \quad \textrm{for } \ell=1, 2, ...,
    \label{Beampattern_Fading_and_Sidelobe}
\end{align}
where $\zeta=0$ for fadings and $\zeta=\frac{1}{2}$ for sidelobes.

\begin{table*}[!t]
	\renewcommand{\arraystretch}{2.8}
	\centering
	\caption{Spatial Properties of the Log-Euclidean, the Conventional and the MVDR beamformers}
 	\label{Table:Properties}
    \begin{tabular}{|c||c|c|c|c|}
   		\hhline{-||----}
        Beamformer &
        \renewcommand{\arraystretch}{1.5}
        \begin{tabular}{@{}cc@{}}
        Spectrum \\(\ref{Analysis_LE_SpatialSpectrum_Final}) 
        \end{tabular}& 
        \renewcommand{\arraystretch}{1.5}
        \begin{tabular}{@{}cc@{}}
        Mainlobe beamwidth \\(\ref{Appendix:LE_BW}) 
        \end{tabular}& 
        \renewcommand{\arraystretch}{1.5}
        \begin{tabular}{@{}cc@{}}
        Peak-to-sidelobe ratio \\(\ref{Appendix:LE_PSLR}) 
        \end{tabular}& 
        \renewcommand{\arraystretch}{1.5}
        \begin{tabular}{@{}cc@{}}
        Power under-estimation due to multipath \\(\ref{Analysis_LE_Sensitivity_to_multipath})  
        \end{tabular}
        \\ 
        \hhline{=::====}	
        Log-Euclidean $P_\text{LE}(\theta)$   &
        $(\sigma_1^2+1)^{b_\theta}-1$ &
        $\displaystyle \frac{1}{\left\| \Dot{\a}_1 \right\|}\sqrt{1\!-\!\frac{\log(\sigma_1^2/2+1)}{\log(\sigma_1^2+1)}}$ &
        $\displaystyle \frac{\sigma_1^2}{(\sigma_1^2\!+\!1)^{b_\text{SL}}\!-\!1}$ &
        $\displaystyle (\sigma_1^2+\sigma_2^2+1)^{\frac{1}{1+\sigma_2^2/\sigma_1^2}}-1$  \\
        \hhline{-||----}	
        Conventional \newline $P_\text{CB}(\theta)$ &
        $(b_\theta\:\sigma_1^2 +1)$ &
        $\displaystyle \frac{1}{\left\| \Dot{\a}_1 \right\|}\sqrt{\frac{1}{2} +\frac{1}{2}\frac{1}{\sigma_1^2} }$ &
        $\displaystyle \frac{ \sigma_1^2+1 }{ b_\text{SL}\:\sigma_1^2+1 }$ &
        $(\sigma_1^2+1)$  \\
        \hhline{-||----}	
        MVDR \newline $P_\text{MV}(\theta)$ &
        $\displaystyle \frac{(\sigma_1^2+1)}{ \sigma_1^2 (1-b_\theta)+1 }$ &
        $\displaystyle \frac{1}{\left\| \Dot{\a}_1 \right\|}\sqrt{\frac{1}{\sigma_1^2}}$ &
        $\displaystyle (1-b_\text{SL})\:\sigma_1^2+1$ &
        $\displaystyle \frac{(\sigma_1^2+\sigma_2^2+1)}{(\sigma_2^2+1)}$ \\
        \hhline{-||----}	
    \end{tabular}
\end{table*}

Back to the general case for further characterization of the proposed LE beamformer, denote by $\theta_1$ the mainlobe direction of an arbitrary spatial spectrum $P(\theta)$,
and by $\theta_\text{BW}$ the direction at which the mainlobe drops to half of its peak value, i.e., $\theta_\text{BW}$ satisfies $P(\theta_\text{BW})=\frac{1}{2}P(\theta_1)$.
\begin{definition}[Half-power beamwidth of a spatial spectrum]
The half-power beamwidth (HPBW) of a spatial spectrum $P(\theta)$ is defined as the angular separation 
\begin{align}
    \text{HPBW} & \triangleq |\theta_\textnormal{BW}-\theta_1|. \label{Appendix:Define_BW} 
\end{align}
\end{definition}
Existing closed-form expressions for the HPBW have been derived for classical beamformers in the \emph{absence} of noise and are therefore independent of the SNR \cite{BibList:Balanis,stoica2005spectral,VanTrees,BibList:Mailloux}. 
For example, it has been reported in \cite{VanTrees} that in the case of ULAs, the HPBW is $\frac{0.89}{(M\!-\!1)d_\lambda}$, which is indeed independent of the SNR. 
HPBW expressions reported in \cite{BibList:Balanis,stoica2005spectral,VanTrees,BibList:Mailloux} actually hold only in the \emph{absence} of noise, which is unrealistic in practice. 
In contrast, we are interested in a HPBW expression that depends on the SNR in order to better illustrate the advantages of the \ac{LE} beamformer in the \emph{presence} of noise. 
Since the definition of $\theta_\text{BW}$ relies on $P(\theta)$, which itself depends on the SNR, then $\theta_\text{BW}$ also depends on the SNR, and so does $b_\text{BW}\triangleq b_{\theta_{BW}}$, defined as the beampattern evaluated in direction $\theta_\textnormal{SL}$. 
Thus, to consider the dependence on the SNR, we recast the HPBW as a function of $b_\text{BW}$. 
\begin{proposition}[Half-power beamwidth as a function of $b_\text{BW}$]
    The first-order approximation of the HPBW defined in (\ref{Appendix:Define_BW}) 
is
\begin{align}
    \text{HPBW} &  = \frac{1}{\left\| \Dot{\a}_1 \right\|}\sqrt{1-b_\textnormal{BW}},
    \label{Appendix:Approx_BW}
\end{align}
where $\Dot{\a}_\theta$ denotes the derivative of $\a_\theta$ with respect to $\theta$. 
\end{proposition}
The proof appears in Appendix \ref{Appendix:LE_Beamwidth_Mainlobe}. We now specialize (\ref{Appendix:Approx_BW}) to the case of the LE beamformer presented in this paper.
\begin{proposition}[Half-power beamwidth of the LE beamformer]
    In the case of the LE beamformer in (\ref{Analysis_LE_SpatialSpectrum_Final}), the HPBW defined in (\ref{Appendix:Approx_BW}) has the following closed-form expression,
    \begin{align}
    \text{HPBW}_\textnormal{LE} & = 
    \frac{1}
    {\left\| \Dot{\a}_1 \right\|}\sqrt{1\!-\!\frac{\log(\sigma_1^2/2\!+\!1)}{\log(\sigma_1^2\!+\!1)}}.
    \label{Appendix:LE_BW}
\end{align}
\end{proposition}
\begin{proof}
Equation (\ref{Analysis_LE_SpatialSpectrum_Final}) yields that $P_\text{LE}(\theta_\text{BW})=(\sigma_1^2+1)^{b_\text{BW}}-1$. Furthermore, we have $P_\text{LE}(\theta_\text{BW})=P_\text{LE}(\theta_1)/2$ by definition. By combining these two expressions for $P_\text{LE}(\theta_\text{BW})$ and recalling that $P_\text{LE}(\theta_1)=\sigma_1^2$, we get after a few manipulations that 
$b_\text{BW} = \frac{\log(\sigma_1^2/2+1)}{\log(\sigma_1^2+1)}$. 
Substituting this result in (\ref{Appendix:Approx_BW}) yields the HPBW expression in (\ref{Appendix:LE_BW}).
\end{proof}
The expression for $b_\text{BW}$ obtained in the above proof depends exclusively on $\sigma_1^2$, thereby on the SNR. 
In turn, the HPBW of the LE beamformer, established in (\ref{Appendix:LE_BW}), also depends on the SNR. 
Note that as the SNR increases, the HPBW decreases.

\begin{definition}[Peak-to-sidelobe ratio of a spatial spectrum]
The peak-to-sidelobe ratio (PSLR) of a spatial spectrum $P(\theta)$ is defined as the following power spectrum ratio
\begin{align}
    \text{PSLR} &\triangleq\frac{P(\theta_1)}{P(\theta_\textnormal{SL})}, \label{Appendix:Approx_PSLR}
\end{align}
where $\theta_\textnormal{SL}$ denotes the direction of the closest sidelobe.
\end{definition}
The PSLR in (\ref{Appendix:Approx_PSLR}) is usually evaluated for classical beamformers in the \emph{absence} of noise \cite{BibList:Balanis,stoica2005spectral,VanTrees,BibList:Mailloux}. For example, the reported PSLR in the case of ULAs is about $13.26$ [dB], which is clearly independent on the SNR. Conversely, in order to exhibit the advantages of the LE beamformer, we propose to evaluate the PSLR in (\ref{Appendix:Approx_PSLR}) in the \emph{presence} of noise, by using a population covariance matrix that accounts for noise.
\begin{proposition}[Peak-to-sidelobe ratio of the LE beamformer]
    In the case of the LE beamformer in (\ref{Analysis_LE_SpatialSpectrum_Final}), the PSLR defined in (\ref{Appendix:Approx_PSLR}) has the following closed-form expression,
    \begin{align}
    \text{PSLR}_\textnormal{LE}&\triangleq\frac{P_\textnormal{LE}(\theta_1)}{P_\textnormal{LE}(\theta_\textnormal{SL})}
   =
   \frac
   {\sigma_1^2}
   {(\sigma_1^2+1)^{b_\textnormal{SL}}-1},
       \label{Appendix:LE_PSLR}
    \end{align}
    where $b_\textnormal{SL}$ is the beampattern evaluated in direction $\theta_\textnormal{SL}$.
\end{proposition}
\begin{proof}
According to (\ref{Analysis_LE_SpatialSpectrum_Final}), we have $P_\text{LE}(\theta_1)=(\sigma_1^2+1)^{b_1}-1$ where $b_1$ denotes the beampattern $b_\theta$ evaluated in direction $\theta_1$. Eq. (\ref{Generic_b}) yields that $b_1=1$, thus $P_\text{LE}(\theta_1)=\sigma_1^2$. Furthermore, it also follows from (\ref{Analysis_LE_SpatialSpectrum_Final}) that $ P_\text{LE}(\theta_\text{SL})=  (\sigma_1^2+1)^{b_\text{SL}}-1$. Substituting these expressions in (\ref{Appendix:Approx_PSLR}) directly yields (\ref{Appendix:LE_PSLR}).
\end{proof}

In order to compare the LE beamformer with the CB and MVDR beamformer, we also evaluate the spectra $P_\text{CB}(\theta)$ and $P_\text{MV}(\theta)$ using the population covariance matrix. 
Then, using (\ref{Appendix:Approx_BW}) and (\ref{Appendix:Approx_PSLR}), we obtain HPBW and PSLR expressions for the \ac{CB} and \ac{MVDR} beamformer. 
The resulting expressions for the power spectra, the HPBW and the PSLR are reported in Table \ref{Table:Properties}. As observed in the case of the LE beamformer, these expressions for the \ac{CB} and \ac{MVDR} beamformer all depend on $\theta$ through $b_\theta$.

\begin{figure*}
\centering
\subfloat[Half-power beamwidth (HPBW).\label{Fig:8_HPBW}]{ \includegraphics[width=0.49\linewidth,trim=4.6cm 10.5cm 5.3cm 10.9cm,clip=true]{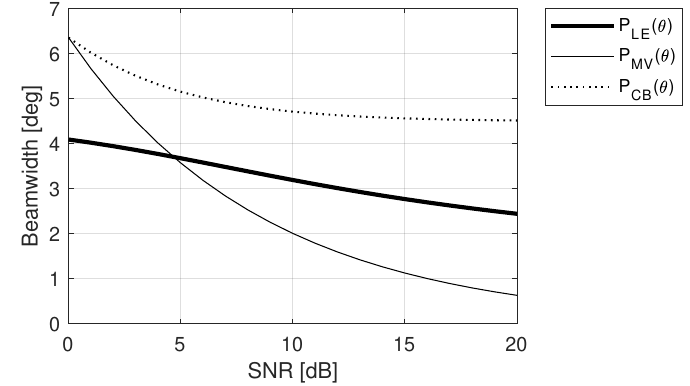}} 
\hfill
\subfloat[Peak-to-sidelobe ratio (PSLR).\label{Fig:8_PSLR}]{ \includegraphics[width=0.49\linewidth,trim=4.6cm 10.5cm 5.3cm 10.9cm,clip=true]{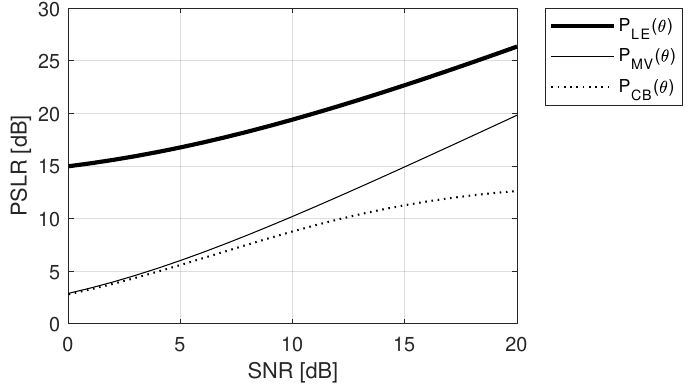}}
\caption{Half-power beamwidth (HPBW) and peak-to-sidelobe ratio (PSLR) evaluated using the population covariance matrix, for the Log-Euclidean beamformer $P_\text{LE}(\theta)$, the MVDR beamformer $P_\text{MV}(\theta)$ and the Conventional beamformer $P_\text{CB}(\theta)$. Case of a ULA with $M=10$ antennas and a single source at boresight.}
\label{Fig:8}
\end{figure*}



The HPBW of the LE, CB and MVDR spectra are presented in Fig. \ref{Fig:8_HPBW} in the case of a ULA with $M=10$ elements and a source in the boresight direction $\theta_1=90$ [deg] for SNR values ranging from 0 to 20 [dB]. The HPBW of the LE spectrum is plotted in Fig. \ref{Fig:8_HPBW} (thick line) together with those of the MVDR spectrum (thin line) and the CB spectrum (dotted line). 
Fig. \ref{Fig:8_HPBW} indicates that for SNR lower than 5 [dB], the thinnest mainlobe is that of the LE spectrum. 
The PSLR of the LE, CB and MVDR spectra are presented in Fig. \ref{Fig:8_PSLR}. 
It appears that the best PSLR is clearly that of the LE spectrum (thick line), towards which the MVDR spectrum (thin line) asymptotically tends to, while that of the CB spectrum (dotted line) converges to the $13.26$ [dB] level reported in \cite{BibList:Balanis,stoica2005spectral,VanTrees,BibList:Mailloux}. 
Fig. \ref{Fig:8} demonstrates the advantages of the LE beamformer over its Euclidean counterparts, namely the CB and the MVDR beamformer, especially at low SNR, in terms of PSLR and HPBW.

Consider now a two-ray multipath channel. Besides the \ac{SoI} in direction $\theta_1$ with incident power $\sigma_1^2$, a fully coherent signal in direction $\theta_2$ with incident power $\sigma_2^2$ impinges on the phased array, thereby corrupting the DoA estimates of the \ac{SoI}.
Denote by $\rho$ the complex scalar with unit magnitude that models the phase difference between the two coherent rays. Then, the population covariance matrix is 
\begin{align}
    \Rb_2 = \left( \sigma_1\a_1 + \rho\sigma_2\a_2 \right)
    \left( \sigma_1\a_1 + \rho\sigma_2\a_2 \right)^H
    +\I.
    \label{Eq:Def_Rh2}
\end{align}
Denote by $P_\text{LE}(\theta|\Rb_2)$, $P_\text{MV}(\theta|\Rb_2)$ and $P_\text{CB}(\theta|\Rb_2)$  the LE, MVDR and CB spectra evaluated in the presence of multipath using $\Rb_2$, instead of $\Rb$ that ignores multipath.
\begin{proposition}[SoI power estimates in multipath conditions]
\label{Proposition_Multipath}
Suppose that $\theta_1$ and $\theta_2$ are sufficiently separated such that $\theta_2$ (the direction of the multipath ray) is away from the mainlobe of the \ac{SoI} propagating from direction $\theta_1$. Then, the power estimates in the \ac{SoI} direction, obtained from the LE, MVDR, and CB spectra, are given by
\begin{align}
    P_\textnormal{LE}(\theta_1|\Rb_2)&\approx(\sigma_1^2+1+\sigma_2^2)^{\frac{\sigma_1^2}{\sigma_1^2+\sigma_2^2}} - 1,
    \label{Analysis_LE_Sensitivity_to_multipath}
    \\
    P_\textnormal{MV}(\theta_1|\Rb_2)&\approx(\sigma_1^2+1+\sigma_2^2)/(\sigma_2^2+1),
    \label{Analysis_MV_Sensitivity_to_multipath}
    \\
    P_\textnormal{CB}(\theta_1|\Rb_2)&\approx(\sigma_1^2+1).
    \label{Analysis_CB_Sensitivity_to_multipath}
\end{align}
\end{proposition}
The proof and details on the approximation appear in Appendix \ref{Appendix:PowerUnderEstimation}. It is noteworthy that the expressions in (\ref{Analysis_LE_Sensitivity_to_multipath}), (\ref{Analysis_MV_Sensitivity_to_multipath}) and (\ref{Analysis_CB_Sensitivity_to_multipath}) are independent of $\theta_1$. 

Ideally, the power estimates in the \ac{SoI} direction should not be affected by multipath, i.e., it should be as close as possible to $\sigma_1^2$ for the LE beamformer and to $\sigma_1^2+1$ for the CB and MVDR beamformer. 
Eq. (\ref{Analysis_CB_Sensitivity_to_multipath}) indicates that this 
ideal robustness is achieved by the \ac{CB}. To illustrate the robustness of the LE beamformer to multipath, the expressions (\ref{Analysis_LE_Sensitivity_to_multipath}), (\ref{Analysis_MV_Sensitivity_to_multipath}) and (\ref{Analysis_CB_Sensitivity_to_multipath}) are numerically evaluated and presented in Fig. \ref{Fig:10} for a grid of values of $\sigma_1^2$ and $\sigma_2^2$ spanning from 0 to 10 [dB].
\begin{figure}
\centering
\includegraphics[width=\linewidth,trim=2cm 0.7cm 1.5cm 2.2cm,clip=true]{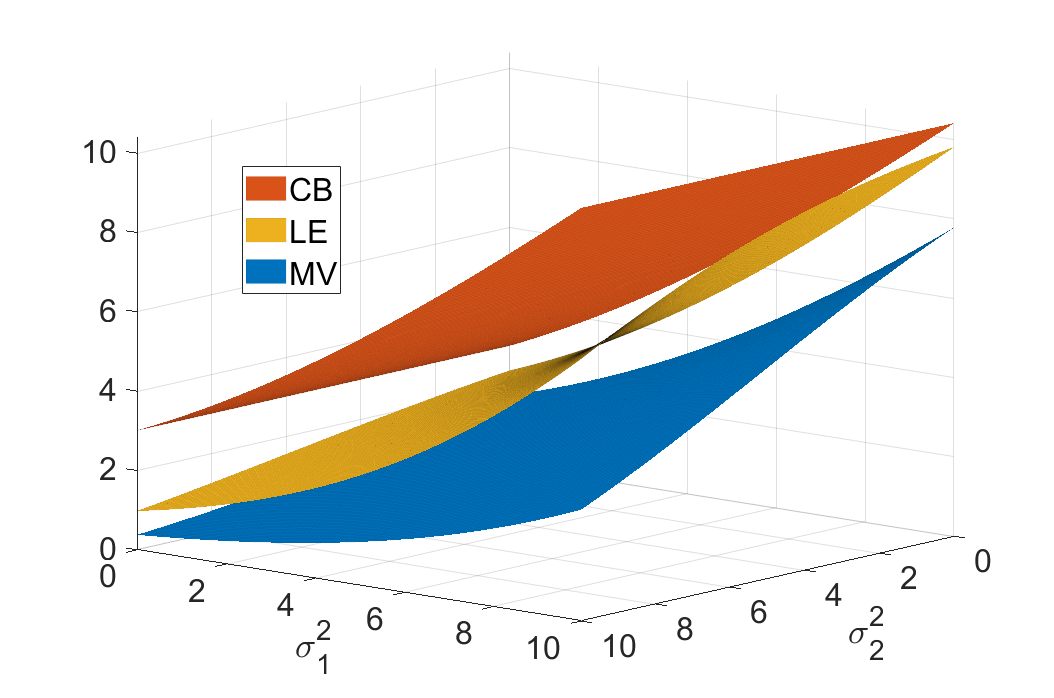}
\caption{The approximate power spectra $P_\text{CB}(\theta_1|\Rb_2)$, $P_\text{LE}(\theta_1|\Rb_2)$, and $P_\text{MV}(\theta_1|\Rb_2)$ given in (\ref{Analysis_LE_Sensitivity_to_multipath}), (\ref{Analysis_MV_Sensitivity_to_multipath}) and (\ref{Analysis_CB_Sensitivity_to_multipath}), as a function of $\sigma_1^2$ and $\sigma_2^2$.}
\label{Fig:10}
\end{figure}
We see in Fig. \ref{Fig:10} that in the presence of multipath, the LE spectrum slightly under-estimates the \ac{SoI} power, outperforming the MVDR spectrum, which significantly under-estimates it. 

\section{Alternative Metrics for Covariance Fitting in DoA Estimation Problems}
\label{Section:Alternative_Riemannian_Metrics}


To complement the analysis and to enhance the difference between the proposed LE beamformer (\ref{LE_PowerSpectrum}), which is a low-complexity approximation of the AI beamformer (\ref{Eq:Estimator_AI}), and the \ac{CB} and MVDR beamformer, we show in this section
that other metrics systematically yield the \ac{CB} and MVDR beamformer, or variants and combinations of which. The considered metrics and their respective spatial spectra are summarized in Table \ref{Table:Comparison}.

\begin{table*}
	\renewcommand{\arraystretch}{1.8}
	\centering
    \caption{Beamformers and DoA-Estimators Obtained by Covariance Fitting for Different Metrics}
 	\label{Table:Comparison}
 	\begin{tabular}{|cc||l|l|l|}
 		\hhline{--||---}
        \multicolumn{2}{|c||}{BEAMFORMER} &
        \multicolumn{1}{c|}{DISTANCE MINIMIZATION} &
        SPATIAL SPECTRUM &
        SHRINKAGE \\
        \hhline{==::===}	
        Conventional &
        $ P_\text{CB}(\theta)$ &
        $ \displaystyle \min_{\sigma^2} d_\text{E}^2( \:\Rh\:,\:\R\: )\:=\:\left\| \Rh-\R \right\|_F^2 $ &
         $\a_\theta^H \Rh \a_\theta$ &
        $g(\lambda)=\lambda$
        \\
        \hhline{-----}
        MVDR (Capon) &
        $ P_\text{MV}(\theta)$ &
        $ \displaystyle \min_{\sigma^2} d_\text{E}^2( \Rh\!{}^{-1}\!\!\!,\R\!{}^{-1}\! )=\left\| \Rh{}^{-1}\!\!\!-\!\R{}^{-1}\! \right\|_F^2 $ &
         $\displaystyle \frac{1}{\a_\theta^H \Rh{}^{-1} \a_\theta}$ &
        $g(\lambda)=-\dfrac{1}{\lambda}$
        \\
        \hhline{-----}
        Log-Euclidean &
        $ P_\text{LE}(\theta)$ &
        $ \displaystyle \min_{\sigma^2} d_\text{LE}^2( \:\Rh\:,\:\R\: )\:=\:
        \left\| \log(\Rh) - \log(\R) \right\|_F^2$ &
        $\!\exp\left(\!\a_\theta^H \!\log(\Rh)\; \a_\theta\!\right)\!-\!\sigma_n^2\!$ &
        $g(\lambda)=\log(\lambda)$
        \\
        \hhline{-----}  
        Affine-Invariant &
        $ P_\text{AI}(\theta)$ &
        $ \displaystyle \min_{\sigma^2} d_\text{AI}^2( \:\Rh\:,\:\R\: )\:=\:
        \left\| \log(\Rh{}^{-\sfrac{1}{2}}\R\Rh{}^{-1/2}) \right\|_F^2
        $ &
        \hspace{1cm} NA &
        \hspace{0.8cm} NA 
        \\
        \hhline{-----}
        Kullback-Leibler (1) &
        $ P_\text{KL}^{(1)}(\theta)$ &
        $ \displaystyle \min_{\sigma^2} d_\text{KL}^2( \:\Rh\:,\:\R\: )=\Trace(\R{}^{-1}\Rh-\I)\!-\!\log\Det(\R{}^{-1}\Rh) $ &
         $\displaystyle \a_\theta^H \Rh \a_\theta - \sigma_n^2$ &
         $g(\lambda)=\lambda$
        \\
        \hhline{-----}
        Kullback-Leibler (2) &
        $ P_\text{KL}^{(2)}(\theta)$ &
        $ \displaystyle \min_{\sigma^2} d_\text{KL}^2( \:\R\:,\:\Rh\: )=\Trace(\Rh{}^{-1}\R-\I)\!-\!\log\Det(\Rh{}^{-1}\R) $ &
         $\displaystyle \frac{1}{\a_\theta^H \Rh{}^{-1} \a_\theta} - \sigma_n^2$ &
         $g(\lambda)=-\dfrac{1}{\lambda}$
        \\
        \hhline{-----}
        Log-Determinant &
        $ P_\text{LD}(\theta)$ &
        $ \displaystyle \min_{\sigma^2} d_\text{LD}^2( \Rh,\R )=
        \log\!\Big(\!\!\Det\Big\{\!\!\frac{\Rh+\R}{2}\!\!\Big\}\!\Big) \!\!-\!\! 
        \frac{1}{2}\log\!\Big(\!\!\Det\big\{\!\Rh\R\!\big\}\!\Big)$ &
        $ \!\!\displaystyle \frac{1}{\a_\theta^H\!(\Rh\!+\!\sigma_n^2\I)^{-1}\a_\theta} \!-\! 2\sigma_n^2\!\!$ &
        $g(\lambda)=-\dfrac{1}{\lambda+\sigma_n^2}$
        \\
        \hhline{--||---}        
    \end{tabular}
\end{table*}

\subsection{Covariance Fitting based on the Kullback-Leibler Distance}
In the Information Geometry framework, the covariance matrices $\R$ and $\Rh$ fully describe the zero-mean Gaussian distribution of the model and the samples. The Kullback-Leibler (KL) divergence is a well-known measure for the distance between two such distributions, its expression is \cite{BibList:Chevallier}
\begin{align}
    d_\text{KL}(\Rh,\R)\triangleq\Trace(\R{}^{-1}\Rh-\I)-\log\Det(\R{}^{-1}\Rh).
    \label{Define_KullbackLeibler}    
\end{align}
The KL divergence is not invariant to inversion and congruence.  
The spatial spectrum obtained by minimizing the KL distance is (see Supplementary Material, Appendix \ref{Appendix:KL})
\begin{align}
     P_\text{KL}^{(1)}(\theta)\triangleq{\displaystyle \argmin_{\sigma^2}}\: d_\text{KL}^2( \Rh,\R )=
    \a_\theta^H \:\Rh\: \:\a_\theta \: - 1,    \label{KL_PowerSpecvrtum_1}
\end{align}
which is similar to the spectrum of the CB beamformer (\ref{Spectrum:CB}) up to the inconsequential offset $-1$.
Since the KL divergence is not symmetric, we get that the spatial spectrum obtained by minimizing 
$d_\text{KL}(\R,\Rh)$ instead of $d_\text{KL}(\Rh,\R)$ is (see Supplementary Material, Appendix \ref{Appendix:KL})
\begin{align}
     P_\text{KL}^{(2)}(\theta)\triangleq{\displaystyle \argmin_{\sigma^2}}\: d_\text{KL}^2( \R,\Rh )=
    {\displaystyle \frac{1}{\a_\theta^H \Rh\!{}^{-1}\! \a_\theta}} - 1,
    \label{KL_PowerSpecvrtum_2}
\end{align}
which coincides with the spectrum of the MVDR beamformer (\ref{Spectrum:MV}) up to a constant offset $-1$. 

\subsection{Covariance Fitting based on the Log-Determinant Metric}
The Log-Determinant (LD) distance between the \ac{HPD} covariance matrices $\Rh$ and $\R$ is defined by \cite{sra2012new}
\begin{align}
    d_\text{LD}^2\left( \!\Rh,\R \!\right)\!\triangleq\!
    \log\!\left(\!\Det\left\{\!\!\frac{\Rh+\R}{2}\!\!\right\}\!\right) \!-\! 
    \frac{\log\!\left(\!\Det\left\{\!\Rh\R\!\right\}\!\right)}{2}.
    \label{LogDet_Distance}
\end{align}
The spatial spectrum obtained by minimizing the LD distance defined in (\ref{LogDet_Distance}) is given by (proof in Supp. Material, App. \ref{Appendix:LD})
\begin{align}
    P_\text{LD}(\theta)=\frac{1}{\a_\theta^H(\Rh+\I)^{-1}\a_\theta} - 2,
    \label{LDPowerSpectrum}
\end{align}
implying that the LD beamformer is a variant of the MVDR beamformer with unit diagonal loading and correction of the power offset.
Furthermore, the LD distance is invariant to matrix inversion and congruence (see Supplementary Material, Appendix \ref{Appendix:LD}). 
As mentioned in \cite[p. 9]{BibList:Horev}, although it approximates the AI metric, the log-Determinant metric is not a Riemannian metric \textit{stricto sensu}.

\subsection{Relation to Shrinkage}

When the sample size is small, a common technique for covariance matrix estimation is to manipulate (``shrink'') the eigenvalues of the sample covariance matrix (e.g., see \cite{stein1986lectures}).
Consider the decomposition of the covariance matrix $\Rh$ into eigenvalues $\lambda_m$ and unit-norm eigenvectors $\u_m$ such that ${\Rh = \sum_{m=1}^M \lambda_m \u_m \u_m^H }$. Then, the spectrum $P_\text{CB}(\theta)$ in (\ref{Spectrum:CB}), the spectrum $P_\text{MV}(\theta)$ in (\ref{Spectrum:MV}), and the spectrum $P_\text{LE}(\theta)$ in (\ref{LE_PowerSpectrum})
can be recast as follows
\begin{align}
    P_\text{CB}(\theta) &= \sum_{m} |\a_\theta^H\u_m|^2 \:\:\lambda_m,  \nonumber \\
    P_\text{MV}(\theta) &= \Big(\sum_{m} |\a_\theta^H\u_m|^2 \:\:\lambda_m^{-1}\:\:\Big)^{-1}, \nonumber \\
    P_\text{LE}(\theta) &= \exp\Big(\sum_{m} |\a_\theta^H\u_m|^2\log(\lambda_m)\Big)-1.
\end{align}
Neglecting constant scalar offsets, these three spectra adhere to a generic spectrum model, given by
\begin{align}
    \hspace{-1.0cm}P(\theta)\triangleq g^{-1}\Big(  \:\sum_{m} |\a_\theta^H\u_m|^2 \: g\big(\lambda_m\big) \: \Big),
    \label{Eq:Proposition_generic_spectrum}
\end{align}
where $g:\mathbb{R}^+\rightarrow\mathbb{R}^+$ is a monotonically increasing function that operates on the eigenvalues of the covariance matrix. The functions $g(\cdot)$ corresponding to the considered beamformers are presented in the rightmost column of Table \ref{Table:Comparison}\footnote{We also claim that the celebrated MUSIC and MODE beamformers are closely related to the generic model (\ref{Eq:Proposition_generic_spectrum}), but it is beyond the scope of this paper.}. 
\subsection{Relation to Existing Riemannian Approaches}
\label{Section:Earlier_Riemannian_Approches}

Define $d_{(M')}^2(\theta)$ from $d_\text{AI}^2(\theta)$ by truncating the summation in (\ref{Eq:AI_model_using_Eigenvalue_sum}) to keep only the $M'$ largest eigenvalues,
\begin{align}
    d_{(M')}^2( \Rh,\R) = \sum_{m=1}^{M'} \log^2\left(\lambda_m\left\{\Rh\!{}^{-\sfrac{1}{2}}\R\Rh\!{}^{-\sfrac{1}{2}}\right\}\right).  \label{Eq:Earlier_Approaches_distances}
\end{align}
In the case of a single signal, the beamformers proposed in \cite{BibList:Coutino,BibList:Dong,BibList:Chahrour1} can be written as
\begin{align}
    P_{[38]}(\theta) &= \frac{1}{d_{(\:1\:)}^2\big(\:\Rh\:,\:\a_\theta\a_\theta^H\big)}
    \nonumber \\
    P_{[39]}(\theta) &= \frac{1}{d_{(\:1\:)}^2\big(\eta\Rh,\:\a_\theta\a_\theta^H\big)}
    \nonumber \\
    P_{[40]}(\theta) &= \frac{1}{d_{(M)}^2\big(\:\Rh\:,\:\a_\theta\a_\theta^H\!+\eta'\I\big)},
\label{Eq:Earlier_Approaches}
\end{align}
where the scaling parameters $\eta$ and $\eta'$ are defined by 
$\eta\triangleq(\hat\sigma_1^2\!+\!\hat\sigma_n^2/M)^{-1}$ and $\eta'\!\triangleq\textrm{std}\big\{\textrm{diag}\{\a_\theta\a_\theta^H\!\}\!\big\}$, and 
where $(\hat\sigma_1^2,\hat\sigma_n^2)$ are Maximum-Likelihood estimates for $(\sigma_1^2,\sigma_n^2)$. 

In Section \ref{Section:Invoking_Riemannian_Metrics}, we outlined advantages of the proposed LE and AI beamformers over the beamformers presented in \cite{BibList:Coutino,BibList:Dong,BibList:Chahrour1}.
Another advantage of the AI and LE beamformers concerns the physical meaning of the spatial spectra.
The AI beamformer $P_\text{AI}(\theta)$, the LE beamformer $P_\text{LE}(\theta)$, and the other beamformers in Table \ref{Table:Comparison} as well, have spatial spectra expressed in power units, the same power units as $\sigma_1^2$, $\sigma_n^2$ and $\Rh$. Conversely, in (\ref{Eq:Earlier_Approaches_distances}) and (\ref{Eq:Earlier_Approaches}), we see that due to the square-logarithm transformation, the units of the spatial spectra $P_{[38]}(\theta)$, $P_{[39]}(\theta)$ and $P_{[40]}(\theta)$ are not power units, but some arbitrary units devoid of physical meaning.
For example, power units should not be arbitrary in order to ensure meaningful PSLR analyses. 
Therefore, the beamformers $P_{[38]}(\theta)$, $P_{[39]}(\theta)$ and $P_{[40]}(\theta)$ cannot be directly considered as power spectra without educated manipulations and heuristic calibrations.

\section{Simulations \& Numerical Examples}

\label{Section:Simulations_and_Numerical_Examples}

In this section, we numerically demonstrate the properties of the proposed \ac{LE} beamformer $P_\text{LE}(\theta)$, whose closed-form expression is given in (\ref{LE_PowerSpectrum}), and compare it to the \ac{CB}, the \ac{MVDR} beamformer, and the Riemannian \ac{AI} beamformer. 

\subsection{Spatial spectrum at low SNR and small sample size}

In the first experiment, the considered phased-array is a \ac{ULA} consisting of $M = 9$ antennas with a half-wavelength antenna separation. This \ac{ULA} intercepts a signal with SNR of 0 [dB] propagating from direction 30 [deg]. The source direction is intentionally chosen away from boresight to add an extra layer of complexity because phased-arrays are less discriminative at endfire. The number of samples is $K=10$, which only critically exceeds the number of antennas. For each sample $k$, the measurement noise $\n_k$ is a zero-mean complex Gaussian vector of length $M$ with covariance matrix $\I$, since the noise power is $\sigma_n^2=1$. The signal $x_k$ is a zero-mean constant modulus signal with phase uniformly drawn in $[0,2\pi]$. 
Figure \ref{Fig:1} presents the spatial power spectrum for the proposed \ac{LE} beamformer $P_\text{LE}(\theta)$ (thick line), where it is compared to the \ac{MVDR} beamformer $P_\text{MV}(\theta)$ (thin line), the \ac{CB} $P_\text{CB}(\theta)$ (dotted line), and to the \ac{AI} spectrum $P_\text{AI}(\theta)$ (dashed line). The true direction of the source is depicted by a black circle.
\begin{figure}
\centering
\includegraphics[width=\linewidth,trim=4.6cm 10.5cm 5.4cm 10.9cm,clip=true]{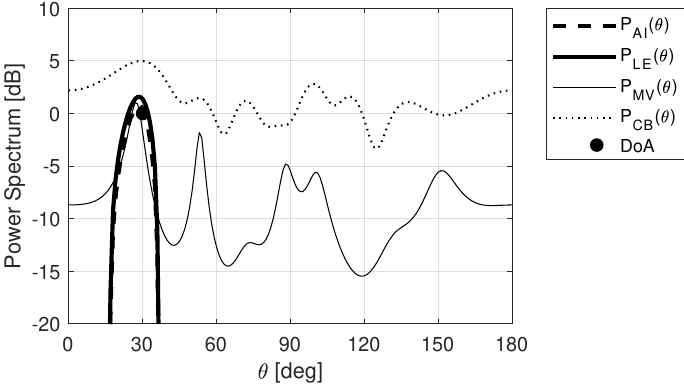}
\caption{The power spectra $P_\text{AI}(\theta)$, $P_\text{LE}(\theta)$, $P_\text{MV}(\theta)$ and $P_\text{CB}(\theta)$. 
The source direction is 30 [deg] and the incident SNR is 0 [dB] (black circle).}
\label{Fig:1}
\end{figure}
We see that the multiple peaks in the \ac{MVDR} spectrum $P_\text{MV}(\theta)$ are misleading and erroneously indicate signals in directions where none exists, e.g., in direction 53 [deg]. The spectrum of the \ac{CB} $P_\text{CB}(\theta)$ fluctuates around the 0 [dB] noise level, and even the actual direction is hardly identifiable. Conversely, the proposed \ac{LE} power spectrum $P_\text{LE}(\theta)$ has a single, unambiguous lobe in the actual DoA.
In addition, we see that the \ac{LE} spectrum is similar to the \ac{AI} spectrum but obtained much more efficiently through a closed-form expression rather than based on an exhaustive, computationally demanding search.

\begin{figure*}
\centering
\subfloat[Four uncorrelated signals (black circles).\label{Fig:2a}]{ \includegraphics[width=0.49\linewidth,trim=4.6cm 10.5cm 5.3cm 10.9cm,clip=true]{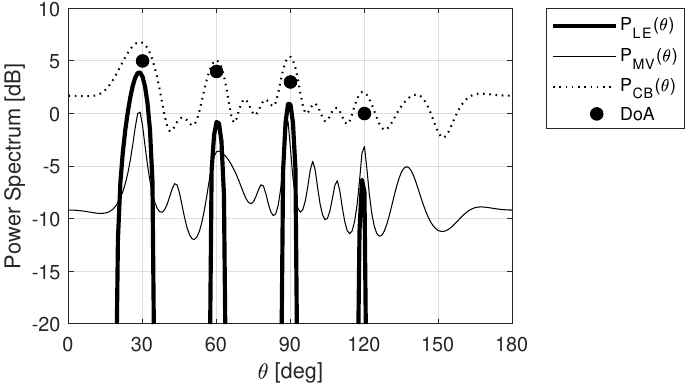}} 
\hfill
\subfloat[Three coherent signals (black circ.) and an uncorrelated signal (white circ.).\label{Fig:2b}]{ \includegraphics[width=0.49\linewidth,trim=4.6cm 10.5cm 5.3cm 10.9cm,clip=true]{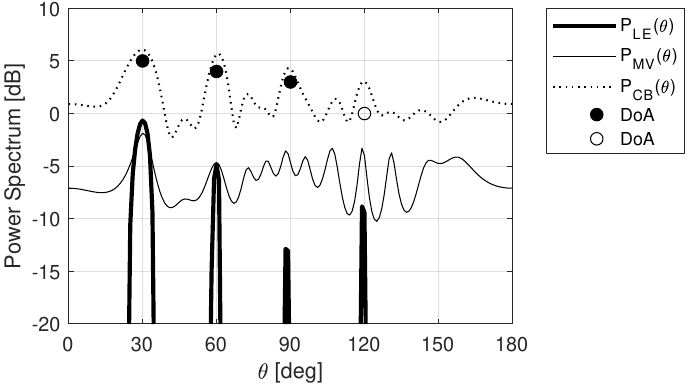}} 
\caption{Power spectra of the proposed \ac{LE} beamformer $P_\text{LE}(\theta)$, the \ac{MVDR} beamformer $P_\text{MV}(\theta)$, and the \ac{CB} $P_\text{CB}(\theta)$.
Four signals propagate from directions 30 [deg], 60 [deg], 90 [deg], and 120 [deg] with respective SNRs: 5 [dB], 4 [dB], 3 [dB], and 0 [dB].}
\label{Fig:2}
\end{figure*}

\begin{figure}
\centering
\includegraphics[width=1.01\linewidth,height=4.4cm,trim=4.6cm 10.5cm 5.3cm 10.9cm,clip=true]{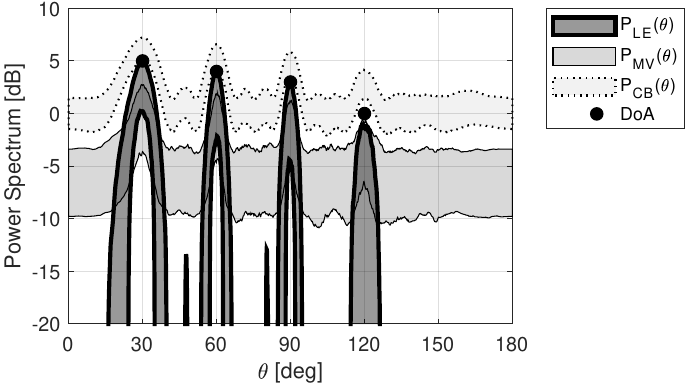} 
\caption{Sleeves of the 5 and 95 percentiles of the Log-Euclidean, the MVDR and the Conventional beamformers.}
\label{Fig:3}
\vspace{-0.5cm}
\end{figure}

In the second experiment, the ULA is composed of $M = 16$ antennas and four uncorrelated signals propagating from directions 30 [deg], 60 [deg], 90 [deg] and 120 [deg]. The SNRs of the incident signals are 5 [dB], 4 [dB], 3 [dB] and 0 [dB], respectively. The number of samples is $K=20$. 
Figure \ref{Fig:2a} compares the spatial power spectrum for the proposed \ac{LE} beamformer $P_\text{LE}(\theta)$ (thick line) with that of the \ac{MVDR} beamformer and the \ac{CB} (thin and dotted lines, respectively). 
We see here as well that the \ac{LE} spectrum $P_\text{LE}(\theta)$ exhibits unambiguous and clearly identifiable lobes at the sources directions (black circles). Conversely, the peaks in the \ac{MVDR} and \ac{CB} spectra are less distinct, and there are peaks in directions where no source exists, e.g., in direction 140 [deg]. 

An even more challenging scenario is considered in Figure \ref{Fig:2b}, in which three of the four signals are fully coherent (marked by black circles), e.g., representing the case of a line-of-sight ray degraded by two multipath rays. Only the fourth signal is kept uncorrelated (marked by a white circle). 
The phase differences between the line-of-sight ray and the two multipath rays, implemented as phasors, are arbitrarily set to $-\frac{2\pi}{3}$ and $+\frac{\pi}{2}$.
We see in the figure that the \ac{LE} spectrum identifies all the DoAs. Conversely, besides the strongest coherent signal at direction 30, the \ac{MVDR} (thin line) fails to identify all the other DoAs, including the uncorrelated one. The \ac{CB} (dotted line) exhibits in the actual DoAs some small lobes emerging above the noise-level fluctuations, but these lobes are clearly inferior to the lobes obtained by the proposed \ac{LE} beamformer (thick line). However, the ``sharpness'' of the \ac{LE} peaks is obtained at the expense of under-estimation of the signal power for each DoA.


We repeat the experiment presented in Fig. \ref{Fig:2a}, with four uncorrelated signals with decreasing powers and DoAs separated by $30$ [deg], $200$ times using i.i.d. realizations of the signals and noise, leading to different sample covariance matrices. The phased-array remains unchanged with $M=16$ elements and the number of samples is set to $K=20$, scarcely more than the number of antennas. 
%
The resulting spatial spectra are presented in Fig. \ref{Fig:3}, where the shades indicate the 5 and 95 percentiles for each of the beamformers.
We see that the \ac{LE} beamformer consists of sharp thin lobes at the directions of the signals, in contrast to the spatial spectra obtained by the other beamformers.

In the third experiment, we consider a dominant signal in direction $\theta_1=30$ [deg] and an uncorrelated weak signal in direction $\theta_2=60$ [deg]. 
The dominant signal SNR is increased from -10 to 50 [dB] while maintaining the weak signal power at 15 [dB] below the dominant signal power. 
The DoAs of the dominant and weak signals are estimated as the directions of the two largest peaks in the spatial power spectrum. The obtained RMSE for $200$ realizations with respect to the dominant signal as a function of the dominant signal SNR is presented in Fig. \ref{Fig:4a}. For reference, we also present the \ac{CRB}.
We see that for a dominant signal SNR from -10 [dB] to 30 [dB], the estimation errors of the dominant signal direction $\theta_1$ obtained by the CB and the LE beamformer are similar and slightly better than that of the MVDR beamformer. Beyond 30 [dB], the estimation error obtained by the CB beamformer increases because of the mainlobe bias induced by the weak signal. 
Conversely, the LE and MVDR errors are low, with a clear advantage for the MVDR beamformer at extremely high SNR. 
In addition, we see that for SNR values from 10 [dB] to 30 [dB], the LE estimator is optimal as it meets the 
\ac{CRB} with respect to $\theta_1$. 
Note that in (unrealistic) high SNR values, the MVDR also meets the bound.

In Fig. \ref{Fig:4b}, we present the root-mean-square error (RMSE) of the weak signal DoA estimates as a function of the weak signal SNR based on $200$ realizations.
Here, we see that the \ac{LE} beamformer outperforms both the \ac{CB} and the MVDR beamformer. 
The MVDR beamformer behaves poorly because the small sample size makes the sample covariance poorly conditioned, and its inverse is then degraded by large errors. 
The \ac{CB} performs even worse because the second largest peak in the CB spectrum is erroneously interpreted as the weak signal direction, but it is in fact one of the dominant signal sidelobes whose directions are $46$ [deg] and $162$ [deg] according to (\ref{Beampattern_Fading_and_Sidelobe}). The resulting RMSE of the \ac{CB} is then 
$\sqrt{(\left(46-\theta_2\right){}^2+\left(162-\theta_2\right){}^2)/2}\approx 72$ [deg]. 
This result is in accordance with the RMSE of the \ac{CB} in Fig. \ref{Fig:4b}.
Conversely, the LE beamformer neither requires inverting the covariance matrix nor suffers from high sidelobes. This experiment demonstrates the advantage of the LE beamformer in recovering DoAs of weak signals. Furthermore, we see that the LE estimation of the weak signal direction is close to optimal as it almost reaches the \ac{CRB}.

\begin{figure*}
\centering
\subfloat[RMSE of DoA estimates for the dominant signal as a function of its SNR.\label{Fig:4a}]{ \includegraphics[width=0.49\linewidth,trim=4.6cm 10.5cm 5.4cm 10.8cm,clip=true]{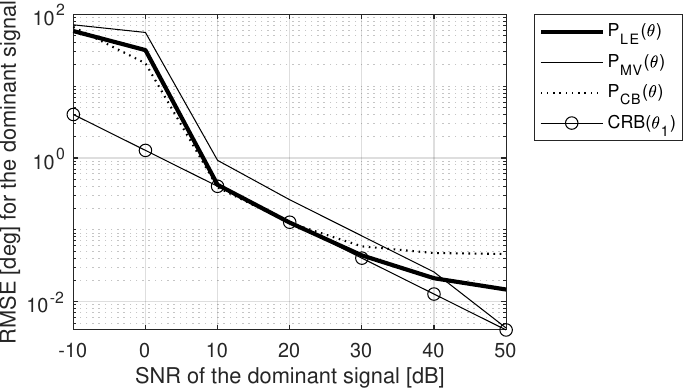}} 
\hfill
\subfloat[RMSE of DoA estimates for the weak signal as a function of its SNR.\label{Fig:4b}]{ \includegraphics[width=0.49\linewidth,trim=4.6cm 10.5cm 5.4cm 10.8cm,clip=true]{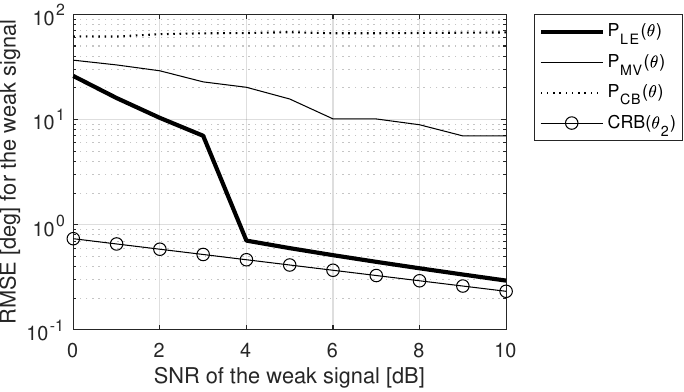}} 
\caption{Root-mean-square error (RMSE) of the directions estimates of the dominant and weak signals, for the \ac{LE}, MVDR and Conventional beamformers (thick line, thin line and dotted line). The SNR of the two signals increase jointly while keeping their power ratio constant and equal to 15[dB].}
\label{Fig:4}
\end{figure*}


\subsection{Demonstrating the spatial properties of the LE beamformer}

We experimentally demonstrate the closed-form expressions of the spatial properties presented in Section \ref{Section:LogEuclideanBeamformer}. The experimental setup includes a phased-array with $M=16$ antennas, and the noise level is set to $\sigma_n^2=1$. Two coherent signals propagate from directions $\theta_1=90$ [deg] and $\theta_2=60$ [deg], and a third signal, uncorrelated with the other two signals, propagates from direction $\theta_3=30$ [deg]. The SNRs of these signals are $\sigma_1^2=5$ [dB], $\sigma_2^2=3$ [dB] and $\sigma_3^2=0$ [dB], respectively. The phase difference between the two coherent signals is $+\frac{\pi}{2}$. The spatial spectra of the \ac{CB}, \ac{LE} and MVDR beamformers, obtained using the population covariance matrix, are represented in Fig. \ref{Fig:5a}. If the multipath ray were absent, the \ac{LE} power spectrum in direction $\theta_1$ would be $\sigma_1^2=5$ [dB]. However, the expression in (\ref{Analysis_LE_Sensitivity_to_multipath}) predicts in the presence of the multipath ray, the power spectrum in direction $\theta_1$ is $3.112$ [dB]. 
Indeed, the obtained power $P_\text{LE}(\theta_1)$ is $3.114$ [dB]. 

To evaluate the predicted beamwidth and sidelobe levels, the last experiment is repeated while keeping only the first signal in direction $\theta_1$, and removing the two other signals. The obtained spatial power spectrum is presented in Fig. \ref{Fig:5b}, focusing on the directions near $\theta_1=90$~[deg]. The mainlobe beamwidth predicted in (\ref{Appendix:LE_BW}) is $2.9$ [deg], while the obtained empirical beamwidth (marked by a horizontal arrow) is $2.5$ [deg]. 
In addition, the sidelobe level predicted in (\ref{Appendix:LE_PSLR}) is $-16.78$ [dB], while the obtained empirical sidelobe (marked by a vertical arrow) is $-16.45$ [deg]. 
We also see in the figure that the \ac{CB} has significantly higher sidelobes, because of noise. However, even in ideal noise-free conditions, the sidelobe level is approximately -13.26 [dB] (see \cite{stoica2005spectral}). Therefore, despite the presence of noise, the \ac{LE} sidelobe level still outperforms the ideal sidelobes of the \ac{CB}.

\begin{figure*}
\centering
\subfloat[Two coherent signals (black circ.) and an uncorrelated signal (white circ.).\label{Fig:5a}]{ \includegraphics[width=0.49\linewidth,trim=4.6cm 10.5cm 5.4cm 10.9cm,clip=true]{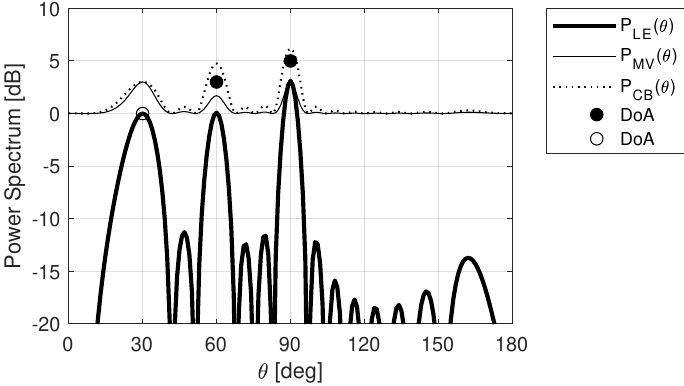}} 
\hfill
\subfloat[A single signal in direction $\theta_1=90\textrm{ [deg]}$. Zoom on the mainlobe.\label{Fig:5b}]{ \includegraphics[width=0.49\linewidth,trim=4.6cm 10.5cm 5.4cm 10.9cm,clip=true]{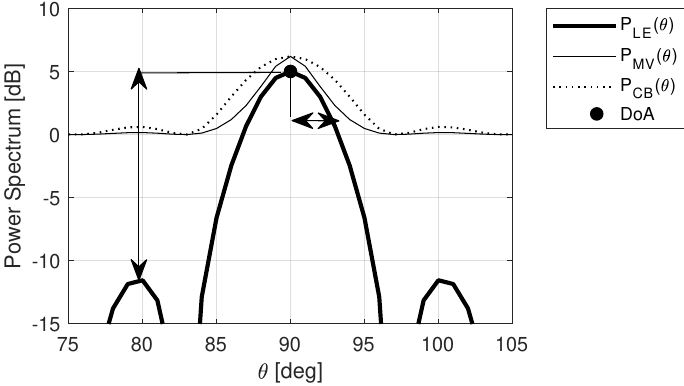}} 
\caption{Power spectrum of the \ac{LE} beamformer $P_\text{LE}(\theta)$, the MVDR beamformer $P_\text{MV}(\theta)$ and the Conventional beamformer $P_\text{CB}(\theta)$.\newline
The power spectra are evaluated using the population covariance matrix.}
\label{Fig:5}
\vspace{-0.3cm}
\end{figure*}

\section{Conclusions} \label{Section:Conclusion}

We introduced a new beamforming approach for \ac{DoA} estimation problems based on \ac{CF} that offers the advantage of exploiting the \ac{HPD} Riemannian geometry of the spatial covariance matrices. 
Following this approach, we presented a new beamformer that relies on the \ac{LE} metric. 
This LE beamformer has a closed-form expression, enabling efficient implementation and analysis of its spatial properties, demonstrating its theoretical advantages over the \ac{CB} and the \ac{MVDR} beamformer at low SNR. 
Numerical experiments illustrate additional advantages of the proposed \ac{LE} beamformer over the \ac{CB} and the \ac{MVDR} beamformer in scenarios that include small sample size with multiple, possibly coherent, signals. These results establish that in the context of \ac{DoA} estimation via \ac{CF}, the \ac{LE} metric, which accounts for the \ac{HPD} geometry of spatial covariance matrices, is a better goodness-of-fit criterion than the Euclidean metric. 
%
%
In future work, we plan to incorporate common techniques of robustification such as diagonal loading in the proposed LE beamformer. The development of robust variants of the LE beamformer will allow us to extend the experimental study and to compare our results to recently-proposed beamformers for DoA estimation.


\appendices

\section{Covariance Fitting Accounting for Noise}\label{Appendix:CovFit_with_noise}

The spectrum $P(\theta)$ obtained by \ac{CF} between $\Rh$ and $\R$ is
\begin{align}
        P(\theta) \!& = \argmin_{\sigma^2} d_\text{E}(\Rh,\R) 
        = \Trace\!\Big\{\!\!\left(\!(\Rh\!-\!\I)\!-\!\sigma^2 \a_\theta \a_\theta^H \!\right)^{\!\!2}\!\Big\}.\! 
\end{align}
Provided that the signal power $\sigma^2$ is not null, the value of $\sigma^2$ that minimizes the latter expression is obtained by matching to $0$ the derivative w.r.t $\sigma^2$, 
\begin{align}
        2\sigma^2 - 2\a_\theta^H(\Rh-\I) \a_\theta = 0.
\end{align}
Rearranging the terms, we get (\ref{Spectrum:CB}) by setting $P(\theta)=\sigma^2$.

Consider now the spectrum $P(\theta)$ obtained by \ac{CF} between $\Rh{}^{-1}$ and $\R^{-1}$.
Note that $\R{}^{-1}=-\alpha\a_\theta \a_\theta^H+\I$ where $\alpha\triangleq\frac{\sigma^2}{\sigma^2+1}$. 
We have
\begin{align}
        P(\theta) & \!=\! \argmin_{\sigma^2} d_\text{E}(\Rh{}^{-1}\!\!\!,\R{}^{-1}) \!=\!
        \Trace\!\Big\{\!\! \left((\Rh{}^{-1}\!\!\!-\!\I)\!+\!\alpha \a_\theta \a_\theta^H \!\right)\!^2\!\Big\} \nonumber\\
        & = \argmin_{\sigma^2} \alpha^2 + 2\alpha\a_\theta^H(\Rh{}^{-1}-\I) \a_\theta.
\end{align}
Provided that the signal power $\sigma^2$ is not null, the value of $\sigma^2$ that minimizes the latter expression is obtained by matching to $0$ the derivative w.r.t $\sigma^2$,
\begin{align}
    \left(2\alpha + 2\a_\theta^H(\Rh{}^{-1}\!\!-\I) \a_\theta\right)\frac{\partial \alpha}{\partial \sigma^2} = 0.
\end{align}
Note that $\frac{\partial \alpha}{\partial \sigma^2}=\frac{1}{(\sigma^2+1)^2}$. 
Inserting the expression for $\alpha$, rearranging the latter equation, and setting $P(\theta)=\sigma^2$ yields
\begin{align}
   P(\theta) = \frac{1}{\a_\theta^H\Rh{}^{-1} \a_\theta} - 1.
\end{align}
The latter expression is equivalent to the MVDR beamformer (\ref{Spectrum:MV}). 
This concludes the proof for our claim in Section \ref{Section:DirectionFinding_by_CovarianceFitting}.

\section{Analysis of the Log-Euclidean Beamformer}\label{Appendix:LE}

\subsection{Spatial Spectrum using the Sample Covariance Matrix}\label{Appendix:LE_Spatial_Spectrum}
The LE power spectrum $P_\text{LE}(\theta)$ defined in (\ref{LE_PowerSpectrum}) consists in minimizing w.r.t. $\sigma^2$ the LE distance between $\Rh$ and $\R$,
\begin{align}
d_\text{LE}^2\left( \Rh,\R \right) &= \left\| \log\Rh - \log\R \right\|_F^2    
\\
&=\Trace\left\{\!
\big(\!\log\Rh\!\big)^2 \!\!-\! 2\log\Rh\log\R \!+\! \big(\!\log\R\!\big)^2
\!\right\}
\nonumber
\end{align}
The value of $\sigma^2$ that minimizes the latter expression is obtained by matching to $0$ the derivative w.r.t $\sigma^2$. We have
\begin{align}
    \frac{\partial d_\text{LE}}{\partial \sigma^2}(\Rh,\R)
    &=\!
    \Trace\!\left\{\!
    - 2\log\Rh\frac{\partial \log\R}{\partial \sigma^2} \!+\!2 \frac{\partial \log\R}{\partial \sigma^2}\log\R
    \!\right\}.
    \nonumber
\end{align}
Matching the latter expression to 0 gives
\begin{align}
    \Trace\left\{ \log\Rh\frac{\partial \log\R}{\partial \sigma^2} \right\}
    =
    \Trace\left\{ \frac{\partial \log\R}{\partial \sigma^2}\log\R \right\}.
    \label{Appendix:LE_matching_to_0}
\end{align}
Note that the eigen-decomposition of the covariance model $\R=\sigma^2\a_\theta\a_\theta^H+\I$ takes the form
\begin{align}
    \R=\bmat{\a_\theta & \U_n} \bmat{\sigma^2+1 & \0b \\ \0b & \I}\bmat{\a_\theta^H \\ \U_n^H}.
\end{align}
where $\U_n$ denotes the eigenvectors associated with the noise subspace. The matrix-logarithm of $\R$ is then
\begin{align}
    \log\R \!=\! \bmat{\a_\theta , \U_n} \bmat{\log(\sigma^2\!+\!1) & \0b \\ \0b & \0b}\bmat{\a_\theta^H \\ \U_n^H}
    \!=\!\log(\sigma^2\!+\!1)\a_\theta\a_\theta^H.
    \nonumber
\end{align}
The derivative of $\log\R$ with respect to $\sigma^2$ is therefore
      $\frac{\a_\theta\a_\theta^H}{\sigma^2+1}$. 
Inserting this result in (\ref{Appendix:LE_matching_to_0}) gives immediately
\begin{align}
    \Trace\left\{ \log\Rh\frac{\a_\theta\a_\theta^H}{\sigma^2+\sigma_n^2} \right\}
    =
    \Trace\left\{ \frac{\a_\theta\a_\theta^H}{\sigma^2+\sigma_n^2}\log\R \right\}.
\end{align}
Rearranging the terms, we get
\begin{align}
    \a_\theta^H\log(\Rh)\a_\theta 
    =
    \a_\theta^H\log(\R)\a_\theta .
\end{align}
Since $\a_\theta$ is the eigenvector of $\R$ associated with the eigenvalue $\log(\sigma^2+1)$, then $\log(\R)\a_\theta=\log(\sigma^2+1)\a_\theta$. As a result, 
\begin{align}
    \a_\theta^H\log(\Rh)\a_\theta 
    =
    \a_\theta^H\left( \log(\sigma^2\!+\!1)\a_\theta \right) 
    =
    \log(\sigma^2\!+\!1).
\end{align}
Applying the exponential function on both sides gives
\begin{align}
    \exp\left(\a_\theta^H\log(\Rh)\a_\theta\right) 
    &=
    \sigma^2+1.
\end{align}
Rearranging the latter equation and setting $P_\text{LE}(\theta)=\sigma^2$, we get the \ac{LE} beamformer
\begin{align}
    P_\text{LE}(\theta) = \exp\left(\a_\theta^H\log(\Rh)\a_\theta\right) - 1.
    \label{Appendix:Analysis_LE_SpatialSpectrum}
\end{align}

\subsection{Spatial Spectrum using the Population Covariance Matrix}
\label{Appendix:LE_Fadings_Sidelobes}
Consider the population covariance matrix $\Rb$ given in (\ref{Eq:CovPopulation_without_interference}) and its eigen-decomposition,
\begin{align}
    \Rb = \bmat{\a_1 , \U_n'}\!\bmat{\sigma_1^2\!+\!1 & \0b \\  \0b  & \I}\!\bmat{\a_1^H \\ \U_n'^H}.
\end{align}
where $\U_n'$ is the matrix of the noise-subspace eigenvectors.
The matrix-logarithm of $\Rb$ is then
\begin{align}
    \log\Rb \!=\! \bmat{\a_1 , \U_n'} \bmat{\log(\sigma_1^2\!+\!1)\!\!\! & \0b \\ \0b & \0b}\bmat{\a_1^H \\ \U_n'^H}
    \!=\!\log(\sigma_1^2\!+\!1)\a_1\a_1^H.
    \nonumber
\end{align}
We now evaluate the operand 
in (\ref{Appendix:Analysis_LE_SpatialSpectrum}),
\begin{align}
    \a_\theta^H\! \log(\Rb) \a_\theta \!
    & = \log(\sigma_1^2\!+\!1)|\a_\theta^H\a_1|^2
    = \log(\sigma_1^2\!+\!1)b_\theta.
    \label{Appendix:Analysis_LE_SpatialSpectrum_LOG}
\end{align}
where the beampattern $b_\theta\triangleq|\a_\theta^H\a_1|^2$ is introduced in (\ref{Generic_b}). 
Inserting (\ref{Appendix:Analysis_LE_SpatialSpectrum_LOG}) in (\ref{Appendix:Analysis_LE_SpatialSpectrum}) gives
\begin{align}
    P_\text{LE}(\theta) 
    & = \exp\left(\log(\sigma_1^2\!+\!1) \: b_\theta \right)\!-\!1
     = (\sigma_1^2\!+\!1)^{b_\theta}\!-\!1.
    \label{Appendix:Analysis_LE_SpatialSpectrum_Final}
\end{align}

\subsection{Derivation of the Half-Power Beamwidth (HPBW)}
\label{Appendix:LE_Beamwidth_Mainlobe}

Consider the second-order Taylor approximation of the beampattern $b_\theta\triangleq |\a_1^H\a_\theta|^2$ in the vicinity of $\theta_1$,
\begin{align}
    b_\theta\approx b_1 
    + \Dot{b}_1(\theta-\theta_1)
    + \Ddot{b}_1 (\theta-\theta_1)^2/ 2,
    \label{Eq:Beamwidth_Taylor_2ndOrder}
\end{align}
where $b_1$, $\Dot{b}_1$ and $\Ddot{b}_1$ denote the beampattern and its derivatives with respect to $\theta$, evaluated at $\theta_1$. One can easily check that
\begin{align}
   \Dot{b}_1 &= \Big.\frac{\partial |\a_1^H\a_\theta|^2}{\partial \theta}\Big|_{\theta_1}
   =2 Re\{\a_1^H\Dot{\a}_1\}, \\
   \Ddot{b}_1 &= \Big.\frac{\partial^2 |\a_1^H\a_\theta|^2}{\partial \theta^2}\Big|_{\theta_1} 
=
2 Re\{\a_1^H\Ddot{\a}_1\}+
2|\a_1^H\Dot{\a}_1|^2.
\end{align}
Provided that the unit-norm steering vector $\a_\theta$ has been properly normalized, the three following identities hold:
\begin{align}
    \a_\theta^H\a_\theta&=1,
    \qquad
    \dot{\a}_\theta^H\a_\theta=0,
    \qquad
    \ddot{\a}_\theta^H\a_\theta=-\left\| \Dot{\a}_\theta \right\|^2,
\end{align}
where $\Dot{\a}_\theta$ and $\Ddot{\a}_\theta$ denote the derivatives of $\a_\theta$ with respect to $\theta$. These three identities evaluated at $\theta_1$ respectively yield
\begin{align}
    b_1 \triangleq |\a_1^H\a_1|^2 = 1, 
    \qquad
    \Dot{b}_1  = 0,
    \qquad
    \Ddot{b}_1 =  -2\left\| \Dot{\a}_1 \right\|^2.
\end{align}
Inserting these results in (\ref{Eq:Beamwidth_Taylor_2ndOrder}) gives
$b_\theta\approx 1 - (\theta-\theta_1)^2 \left\| \Dot{\a}_1 \right\|^2$. 
Rearranging the terms, we get $|\theta-\theta_1|  \approx \sqrt{1-b_\theta} / \left\| \Dot{\a}_1 \right\|$. 
The latter result, evaluated at $\theta=\theta_\text{BW}$, achieves the proof of (\ref{Appendix:Approx_BW}).

Note that in the case of a ULA with a large number of elements, when the SNR $\sigma_1\rightarrow\infty$, the expression (\ref{Appendix:Approx_BW}) approaches the classical expressions for HPBW, e.g. $\frac{0.89}{Md_\lambda}$. Indeed, since $\left\| \Dot{\a}_\theta \right\|^2\!=\!\frac{1}{12}\left(2\pi d_\lambda\right)^2\left(M^2-1\right)$ at boresight, then the \emph{two-sided} HPBW of the \ac{CB} spectrum given in Table \ref{Table:Properties} satisfies
\begin{align}
HPBW= \frac{2}{\left\| \Dot{\a}_1 \right\|}\sqrt{\frac{1}{2} +\frac{1}{2}\frac{1}{\sigma_1^2} } 
\rightarrow
\frac{0.78}{Md_\lambda} \approx \frac{0.89}{Md_\lambda}
\end{align}

\section{Power Under-Estimation due to Multipath} \label{Appendix:PowerUnderEstimation}

\subsection{Log-Euclidean Beamformer: Proof for (\ref{Analysis_LE_Sensitivity_to_multipath})}

Assume that the multipath in direction $\theta_2$ is not within the beampattern of the mainlobe centered on $\theta_1$. Then $\a_1^H\a_2$ has a small magnitude, and the eigen-decomposition of the population matrix $\Rb_2$ given in (\ref{Eq:Def_Rh2}) in the case of two correlated rays takes the form 
\begin{align}
    \Rb_2 &\approx \bmat{\u_s & \U_n''}\bmat{\sigma_1^2+\sigma_2^2+1 & \0b\\ \0b & \I}\bmat{\u_s & \U_n''}^H,
\end{align}
where the eigenvector $\u_s$ satisfies 
$\u_s=\frac{\sigma_1\a_1+\rho\sigma_2\a_2}{\left\|\sigma_1\a_1+\rho\sigma_2\a_2\right\|}$.
Then, the matrix-logarithm of $\Rb_2$ is
$\log\Rb_2 \!=\!\log(\sigma_1^2\!+\!\sigma_2^2\!+\!1)\u_s\u_s^H$. 
Since $\a_1^H\a_2$ has a small magnitude, then
\begin{align}
    \a_1^H\log\Rb_2\a_1 &=\log(\sigma_1^2+\sigma_2^2+1)\left|\a_1^H\u_s\right|^2, \\
    &\approx \log(\sigma_1^2+\sigma_2^2+1)\frac{\sigma_1^2}{\sigma_1^2+\sigma_2^2}.
\end{align}
Then, the power spectrum in direction $\theta_1$ is
\begin{align}
     P_\text{LE}(\theta_1|\Rb_2)&=
    \exp(\a_1^H \log(\Rb_2)\; \a_1)-1,
    \\
    &=
    (\sigma_1^2+\sigma_2^2+1)^{\frac{\sigma_1^2}{\sigma_1^2+\sigma_2^2}}-1.
\end{align}
which concludes the proof of the expression in (\ref{Analysis_LE_Sensitivity_to_multipath}).

\subsection{MVDR Beamformer: Proof for (\ref{Analysis_MV_Sensitivity_to_multipath})}
Consider the population covariance matrix $\Rb_2$ in the case of two coherent signals given in (\ref{Eq:Def_Rh2}).
\begin{align}
    \Rb_2=(\sigma_1\a_1 + \rho\sigma_2\a_2)(\sigma_1\a_1 + \rho\sigma_2\a_2)^H+\I.
\end{align}
The Searle identity \cite{BibList:Bernstein} yields a closed-form expression for the inverse of $\Rb_2$,
\begin{align}
    \Rb_2^{-1}=-\frac{(\sigma_1\a_1 + \rho\sigma_2\a_2)(\sigma_1\a_1 + \rho\sigma_2\a_2)^H}{1+\left\| \sigma_1\a_1 + \rho\sigma_2\a_2 \right\|^2}
    +\I.
\end{align}
Left- and right-multiplying by $\a_1$ gives
\begin{align}
    \a_1^H\Rb_2^{-1}\a_1=-\frac{\left| \sigma_1 + \rho\sigma_2\a_1^H\a_2 \right|^2}{1+\left\| \sigma_1\a_1 + \rho\sigma_2\a_2 \right\|^2}
    +1.
\end{align}
Since $|\a_1^H\a_2|\ll 1$, then
\begin{align}
    \a_1^H\Rb_2^{-1}\a_1\approx -\frac{\sigma_1^2}{1+\sigma_1^2 + \sigma_2^2}
    +1
    =
    \frac{1 + \sigma_2^2}{1+\sigma_1^2 + \sigma_2^2}.
\end{align}
The inverse of the left-hand side expression is identified as the MVDR spectrum in direction $\theta_1$. Therefore, 
\begin{align}
    P_\text{MV}(\theta_1|\Rb_2)
    \approx
    \frac{1+\sigma_1^2 + \sigma_2^2}{1 + \sigma_2^2}.
\end{align}

\subsection{Conventional Beamformer: Proof for (\ref{Analysis_CB_Sensitivity_to_multipath})}
Consider the population covariance matrix $\Rb_2$ in the case of two coherent signals given in (\ref{Eq:Def_Rh2}). Then, the power spectrum in direction $\theta_1$, given by $P_\text{CB}(\theta_1|\Rb_2)=
    \a_1^H \Rb_2\: \a_1$, satisfies
\begin{align}
    P_\text{CB}(\theta_1|\Rb_2)&=
    \a_1^H
    \left( \sigma_1\a_1 + \rho\sigma_2\a_2 \right)
    \left( \sigma_1\a_1 + \rho\sigma_2\a_2 \right)^H\a_1
    +1
    \nonumber
    \\
    &=
    \sigma_1^2+2\Real\left\{\a_1^H\a_2\sigma_1\sigma_2\rho\right\} + |\a_1^H\a_2|^2\sigma_2^2+1. 
    \nonumber
\end{align}
Assume that the multipath in direction $\theta_2$ is not within the beampattern of the mainlobe centered on $\theta_1$. Then $\a_1^H\a_2$ has a small magnitude, and we get $P_\text{CB}(\theta_1|\Rb_2)\approx\sigma_1^2+1$.



\bibliographystyle{IEEEtran}
\bibliography{main}

\clearpage
\vspace*{1cm}
\section*{\Huge Supplementary Material}
\vspace*{1cm}

\section{Analysis of the CB Beamformer}\label{Appendix:Analysis_CB}

\subsection{Fadings and Sidelobes of the Spatial Spectrum}

Consider the Conventional Beamformer $P_\text{CB}(\theta) = \a_\theta^H \Rb \a_\theta$. 
According to (\ref{Eq:CovPopulation_without_interference}), the population covariance matrix in the case of a single emitter is $\Rb = \sigma_1^2 \a_1 \a_1^H +\I$, therefore
\begin{align}
    P_\text{CB}(\theta) 
    & = \a_\theta^H (\sigma_1^2\a_1\a_1^H+\I) \a_\theta \\
    & = \sigma_1^2\a_\theta^H \a_1\a_1^H\a_\theta+1  \\
    & = \sigma_1^2 b_\theta+1,
    \label{Appendix:Analysis_CB_SpatialSpectrum}
\end{align}
where $b_\theta\triangleq\left|\a_\theta^H\a_1\right|^2$ is the beampattern introduced in (\ref{Generic_b}). 
Since the spectrum is a linear function of $b_\theta$, the fadings and the sidelobes of $P_\text{CB}(\theta)$ occur at the same directions than the fadings and the sidelobes of the beampattern $b_\theta$.

\subsection{Derivation of the Beamwidth}

According to (\ref{Appendix:Analysis_CB_SpatialSpectrum}), the CB spectrum in the actual direction $\theta_1$ and the half-power direction $\theta_\text{BW}$ are 
\begin{align}
    P_\text{CB}(\theta_\text{BW})&=\sigma_1^2 b_\text{BW}+1,
    \label{Appendix:P_CB_BW}
    \\
    P_\text{CB}(\theta_1)&=\sigma_1^2+1.
    \label{Appendix:P_CB_peak}    
\end{align}
By definition, we have
$    P_\text{CB}(\theta_\text{BW})=\frac{1}{2}P_\text{CB}(\theta_1)$. 
Therefore, 
    $\sigma_1^2 b_\text{BW}+1=
    \frac{1}{2}\left(
    \sigma_1^2+1
    \right)$. 
Rearranging the terms, we get
\begin{align}
     b_\text{BW}=
    \frac{1}{2}\left(1 -\frac{1}{\sigma_1^2} \right).
\end{align}
Inserting this result in the expression for $BW$ in (\ref{Appendix:Approx_BW}) gives
\begin{align}
    BW & \approx \frac{1}
    {\left\| \Dot{\a}_1 \right\|}\sqrt{1-b_\text{BW}}  =\frac{1}
    {\left\| \Dot{\a}_1 \right\|}\sqrt{\frac{1}{2} +\frac{1}{2}\frac{1}{\sigma_1^2} }.
    \label{Appendix:CB_BW}
\end{align}

\subsection{Derivation of the PSLR}

The CB spectrum values in the actual direction $\theta_1$ and the sidelobe direction $\theta_\text{SL}$ are:
\begin{align}
    P_\text{CB}(\theta_1)&=\sigma_1^2 \left|\a_1^H\a_1\right|^2 + 1 = \sigma_1^2+1,
    \\
    P_\text{CB}(\theta_\text{SL})&=\sigma_1^2 \left|\a_1^H\a_{SL}\right|^2 + 1 = \sigma_1^2 b_\text{SL}+1.
\end{align}
where $b_\text{SL}\triangleq|\a_1^H\a_{SL}|^2$. 
Therefore, the PSLR is
\begin{align}
    PSLR_\text{CB}&\triangleq\frac{P_\text{CB}(\theta_1)}{P_\text{CB}(\theta_\text{SL})}
   =
    \frac{ \sigma_1^2+1 }{ \sigma_1^2 b_\text{SL}+1 }.
\end{align}

\IEEEpubidadjcol 

\mbox{}\newline
\section{Analysis of the MV Beamformer}\label{Appendix:Analysis_MV}

\subsection{Fadings and Sidelobes of the Spatial Spectrum}

Consider the MVDR beamformer $P_\text{MV}(\theta) = \frac{1}{\a_\theta^H \Rb^{-1} \a_\theta}$. 
Note that the inverse of the population covariance matrix $\Rb=\sigma_1^2\a_1\a_1^H\!+\!\I$ is given by 
$\Rb{}^{-1}\!=-\frac{\sigma_1^2}{\sigma_1^2+1}\a_1\a_1^H\!+\!\I$. Then,
\begin{align}
    P_\text{MV}(\theta) 
    & = \frac{1}{\a_\theta^H \left(-\frac{\sigma_1^2}{\sigma_1^2+1}\a_1\a_1^H+\I\right) \a_\theta} 
    \\
    & =
    \frac{1}{ -\frac{\sigma_1^2}{\sigma_1^2+1}b_\theta + 1 } 
     =
    \frac{\sigma_1^2+1}{ \sigma_1^2 (1-b_\theta)+1 }, 
    \label{Appendix:Analysis_MV_SpatialSpectrum}
\end{align}
where $b_\theta\triangleq\left|\a_\theta^H\a_1\right|^2$ is introduced in (\ref{Generic_b}). 
Since the spectrum increases with $b_\theta$, the fadings and the sidelobes of $P_\text{MV}(\theta)$ occur at the same directions than the fadings and the sidelobes of the beampattern $b_\theta$.

\subsection{Derivation of the Beamwidth}

According to (\ref{Appendix:Analysis_MV_SpatialSpectrum}), the MVDR spectrum in the actual direction $\theta_1$ and the half-power direction $\theta_\text{BW}$ are 
\begin{align}
    P_\text{MV}(\theta_\text{BW})&=\frac{\sigma_1^2+1}{ \sigma_1^2 (1-b_\text{BW})+1 },
    \label{Appendix:P_MV_BW}
    \\
    P_\text{MV}(\theta_1)&=\sigma_1^2+1.
    \label{Appendix:P_MV_peak}
\end{align}
By definition, we have $P_\text{MV}(\theta_\text{BW})=\frac{1}{2}P_\text{MV}(\theta_1)$. Therefore,
\begin{align}
    \frac{\sigma_1^2+1}{ \sigma_1^2 (1-b_\text{BW})+1 }=
    \frac{\sigma_1^2+1}{2}.
\end{align}
Rearranging the terms, we get $b_\text{BW} = 1-\frac{1}{\sigma_1^2}$. 
Inserting this result in the expression for the beamwidth established in (\ref{Appendix:Approx_BW}) gives
\begin{align}
    BW & \approx \frac{1}
    {\left\| \Dot{\a}_1 \right\|}\sqrt{1-b_\text{BW}}  =
    \frac{1}
    {\left\| \Dot{\a}_1 \right\|}\frac{1}{\sqrt{\sigma_1^2}}.
    \label{Appendix:MV_BW}
\end{align}

\subsection{Derivation of the PSLR}

The MV spectrum values in the actual direction $\theta_1$ and the sidelobe direction $\theta_\text{SL}$ are:
\begin{align}
    P_\text{MV}(\theta_1)&=\frac{\sigma_1^2+1}{\sigma_1^2(1-\left|\a_1^H\a_1\right|^2)+1}
    =
    \sigma_1^2+1,
    \\
    P_\text{MV}(\theta_\text{SL})&
    =
    \frac{\sigma_1^2+1}{\sigma_1^2(1-b_\text{SL})+1}.
\end{align}
where $b_\text{SL}\triangleq|\a_1^H\a_{SL}|^2$. 
Therefore, the PSLR is
\begin{align}
    PSLR_\text{MV}&\triangleq\frac{P_\text{MV}(\theta_1)}{P_\text{MV}(\theta_\text{SL})}
   =
   \sigma_1^2(1-b_\text{SL})+1.
\end{align}

\mbox{}\newline
\section{Analysis of the Kullback-Leibler Beamformer}\label{Appendix:KL}

Consider the model (\ref{Eq:CovModel}) for the covariance matrix. 
In this Appendix, we will invoke the following identities
\begin{align}
    &\R{}^{-1}\a_\theta=\frac{1}{\sigma^2+1}\a_\theta,
    \\
    &\partial \M{}^{-1} = -\M{}^{-1} (\partial \M) \M{}^{-1},
    \\
    &\partial\log\Det(\M) = \Trace\left\{\M^{-1}(\partial\M)\right\}.
\end{align}
The Kullback-Leibler (KL) power spectrum $P_\text{KL}^{(1)}(\theta)$ defined in (\ref{KL_PowerSpecvrtum_1}) consists in minimizing w.r.t. $\sigma^2$ the KL divergence between $\Rh$ and $\R$ defined by
\begin{align}
    d_\text{KL}(\Rh,\R)\triangleq\Trace\Big(\R{}^{-1}\Rh-\I\Big)-\log\Det(\R{}^{-1}\Rh).
\end{align}
The value of $\sigma^2$ that minimizes the latter expression is obtained by matching to $0$ the derivative w.r.t $\sigma^2$. We have
\begin{align}
    &\frac{\partial d_\text{KL}}{\partial \sigma^2}(\Rh,\R)=\Trace\Big(\frac{\partial\R{}^{-1}}{\partial\sigma^2}\Rh\Big)-\frac{\partial\log\Det(\R{}^{-1}\Rh)}{\partial \sigma^2}
    \\
    &=+\Trace\Big(\frac{\partial\R{}^{-1}}{\partial\sigma^2}\Rh\Big)
    \!-\!\Trace
    \left\{\!(\R{}^{-1}\Rh)^{-1}\frac{\partial (\R{}^{-1}\Rh)} {\partial \sigma^2}\right\}
    \nonumber\\
    &=+\Trace\Big(\frac{\partial\R{}^{-1}}{\partial\sigma^2}\Rh\Big)
    \!-\!\Trace
    \left\{\!(\Rh{}^{-1}\R)\frac{\partial \R{}^{-1}} {\partial \sigma^2}\Rh\right\}
    \nonumber\\
    &=+\Trace
    \left\{\!(\I-\Rh{}^{-1}\R)\frac{\partial \R{}^{-1}} {\partial \sigma^2}\Rh\right\}
    \nonumber\\
    &=-\Trace
    \left\{\!(\I-\Rh{}^{-1}\R)\R{}^{-1}\frac{\partial \R} {\partial \sigma^2}\R{}^{-1}\Rh\right\}
    \nonumber\\
    &=-\Trace
    \left\{\!(\I-\Rh{}^{-1}\R)\R{}^{-1}\a_\theta\a_\theta^H\R{}^{-1}\Rh\right\}.
    \nonumber   
\end{align}
Rearranging the terms in the trace, we immediately get
\begin{align}
    &\frac{\partial d_\text{KL}}{\partial \sigma^2}(\Rh,\R)=-\a_\theta^H\R{}^{-1}\Rh(\I-\Rh{}^{-1}\R)\R{}^{-1}\a_\theta
    \\    
    &=-(\a_\theta^H\R{}^{-1}\Rh\R{}^{-1}\a_\theta-\a_\theta^H\R{}^{-1}\a_\theta)
    \nonumber\\
    &=-\frac{1}{(\sigma^2+1)^2}\a_\theta^H\Rh\a_\theta+\frac{1}{\sigma^2+1}.
    \label{Appendix_KL_diff_1}
\end{align}
where the last equality results from the first aforementioned identity. Matching to zero, we get
\begin{align}
    \sigma^2
   =
   \a_\theta^H\Rh\a_\theta - 1,
\end{align}
which achieves the proof of (\ref{KL_PowerSpecvrtum_1}).

The power spectrum (\ref{KL_PowerSpecvrtum_2}) consists in minimizing w.r.t. $\sigma^2$ the KL divergence between $\R$ and $\Rh$ defined by
\begin{align}
    d_\text{KL}(\R,\Rh)\triangleq\Trace\Big(\Rh{}^{-1}\R-\I\Big)-\log\Det(\Rh{}^{-1}\R).
\end{align}
The value of $\sigma^2$ that minimizes the latter expression is obtained by matching to $0$ the derivative w.r.t $\sigma^2$. We have
\begin{align}
    &\frac{\partial d_\text{KL}}{\partial \sigma^2}(\R,\Rh)=\Trace\Big(\Rh{}^{-1}\frac{\partial\R}{\partial\sigma^2}\Big)-\frac{\partial\log\Det(\Rh{}^{-1}\R)}{\partial \sigma^2}
    \\
    &= +\Trace\Big(\Rh{}^{-1}\frac{\partial\R}{\partial\sigma^2}\Big)-\Trace\left\{
    (\Rh{}^{-1}\R)^{-1}
    \frac{\partial(\Rh{}^{-1}\R)}{\partial \sigma^2}\right\}
    \nonumber\\
    &= +\Trace\Big(\Rh{}^{-1}\frac{\partial\R}{\partial\sigma^2}\Big)-\Trace\left\{
    (\R{}^{-1}\Rh)\Rh{}^{-1}
    \frac{\partial\R}{\partial \sigma^2}\right\}
    \nonumber\\
    &= \Trace\left\{
    (\I-\R{}^{-1}\Rh)\Rh{}^{-1}
    \frac{\partial\R}{\partial \sigma^2}\right\}
    \nonumber\\
    &= \Trace\left\{
    (\I-\R{}^{-1}\Rh)\Rh{}^{-1}
    \a_\theta\a_\theta^H\right\}.
\end{align}
Rearranging the terms in the trace, we immediately get
\begin{align}
    \frac{\partial d_\text{KL}}{\partial \sigma^2}(\R,\Rh)&= 
    \a_\theta^H(\I-\R{}^{-1}\Rh)\Rh{}^{-1}\a_\theta
    \nonumber\\
    &=
    \a_\theta^H\Rh{}^{-1}\a_\theta-\a_\theta^H\R{}^{-1}\a_\theta
    \nonumber\\
    &=\a_\theta^H\Rh{}^{-1}\a_\theta-\frac{1}{\sigma^2+1},
    \label{Appendix_KL_diff_2}
\end{align}
where the last equality results from the first aforementioned identity. Matching to zero, we get
\begin{align}
    \sigma^2=\frac{1}{\a_\theta^H\Rh{}^{-1}\a_\theta}-1,
\end{align}
which achieves the proof of (\ref{KL_PowerSpecvrtum_2}).


\mbox{}\newline
\section{Analysis of the Log-Determinant Beamformer}\label{Appendix:LD}

The Log-Determinant (LD) power spectrum $P_\text{LD}(\theta)$ defined in (\ref{LDPowerSpectrum}) 
consists in minimizing w.r.t. $\sigma^2$ the LD distance between $\Rh$ and $\R$ defined by
\begin{align}
    d_\text{LD}^2\!\left( \!\Rh,\!\R \!\right)\!\triangleq\!
    \log\!\left(\!\Det\left\{\!\!\frac{\Rh\!+\!\R}{2}\!\!\right\}\!\right) \!-\! 
    \frac{\log\!\left(\!\Det\left\{\!\Rh\R\!\right\}\!\right)}{2}.
\end{align}
The value of $\sigma^2$ that minimizes $d_\text{LD}^2\left( \Rh,\R \right)$ is obtained by matching to $0$ the derivative w.r.t $\sigma^2$. We immediately get
\begin{align}
    \frac{\partial}{\partial\sigma^2}\log\!\left(\!\!\Det\left\{\!\frac{\Rh\!+\!\R}{2}\!\right\}\!\!\right)
    \!\!-\!\!
    \frac{1}{2}\frac{\partial}{\partial\sigma^2}\log\!\left(\!\Det\{\Rh\R\}\!\right)
    \!=\!0.
    \label{Eq:LogDet_Derivative}
\end{align}
Note that the log-determinant derivative satisfies the identity $\frac{\partial}{\partial\sigma^2}\log\left(\Det\{\M\}\right)=\Trace\left\{\M^{-1}\frac{\partial\M}{\partial\sigma^2}\right\}$. 
Therefore
\begin{align}
    \!\!\!\!
    \begin{cases}
        \frac{\partial}{\partial\sigma^2}\log\!\left(\Det\left\{\frac{\Rh+\R}{2}\right\}\right)=    \Trace\left\{\left(\Rh+\R\right)^{-1}\frac{\partial\R}{\partial\sigma^2}\right\} 
        \\
        \frac{\partial}{\partial\sigma^2}\log\!\left(\Det\{\Rh\R\}\right)=\Trace\left\{(\Rh\R)^{-1}\Rh\frac{\partial(\R)}{\partial\sigma^2}\right\}
    \end{cases}
\end{align}
Since the derivative of $\R$ w.r.t. $\sigma^2$ is $\a_\theta\a_\theta^H$, we get after few manipulations
\begin{align}
    \!\!\!\!\!\!
    \begin{cases}
        \frac{\partial}{\partial\sigma^2}\log\!\left(\Det\left\{\frac{\Rh+\R}{2}\right\}\right)=    \a_\theta^H\left(\Rh+\R\right)^{-1}\a_\theta
        \\
        \frac{\partial}{\partial\sigma^2}\log\!\left(\Det\{\Rh\R\}\right)=\a_\theta^H\R^{-1}\a_\theta
    \end{cases}
\end{align}
Then, equation (\ref{Eq:LogDet_Derivative}) yields
\begin{align}
    \a_\theta^H(\Rh+\R)^{-1}\a_\theta
    -
    \frac{1}{2}\a_\theta^H\R^{-1}\a_\theta
    =0. 
\end{align}
Inserting the expression for $\R$ and $\R{}^{-1}$, the latter equation reduces to
\begin{align}
    \a_\theta^H(\Rh+\sigma^2\a_\theta\a_\theta^H+\I)^{-1}\a_\theta
    =
    \frac{1}{2}\frac{1}{\sigma^2+1}.
\end{align}
Invoking the Sherman-Morrison identity, we get after few manipulations,
\begin{align}
    \a_\theta^H\!(\Rh+\I)^{-1}\a_\theta -  \frac{ \sigma^2\!\left(\!\a_\theta^H\!(\Rh\!+\!\I)^{-1} \a_\theta\!\right)^2}{1\!+\!\sigma^2\a_\theta^H\!(\Rh\!+\!\I)^{-1}\a_\theta}
    =
    \frac{1}{2}\frac{1}{\sigma^2\!+\!1}.
\end{align}
Simplifying the left-hand terms yields
\begin{align}
    \frac{\a_\theta^H(\Rh+\I)^{-1}\a_\theta}{1+\sigma^2\a_\theta^H(\Rh+\I)^{-1}\a_\theta}
    =
    \frac{1}{2}\frac{1}{\sigma^2+1}.
\end{align}
Rearranging the latter equation and setting $P_\text{LD}(\theta)=\sigma^2$, we get the Log-Determinant beamformer
\begin{align}
    P_\text{LD}(\theta) = \frac{1}{\a_\theta^H(\Rh+\I)^{-1}\a_\theta} - 2.
\end{align}
It achieves the proof of (\ref{LDPowerSpectrum}).

We now check that the Log-Determinant metric is indeed invariant to matrix inversion. To that end, we evaluate:
\begin{align}
    &d_\text{LD}^2\left( \!\Rh{}^{-1}\!\!,\R{}^{-1} \!\right)\!=\!\nonumber
    \\
    &\log\!\left(\!\Det\left\{\!\!\frac{\Rh{}^{-1}\!\!+\!\R{}^{-1}}{2}\!\!\right\}\!\right) - 
    \frac{\log\!\left(\!\Det\left\{\!\Rh{}^{-1}\R{}^{-1}\!\right\}\!\right)}{2}=
    \\
    &\log\!\left(\!\Det\left\{\!\!\frac{\Rh{}^{-1}(\Rh\!\!+\!\R)\R{}^{-1}}{2}\!\!\right\}\!\right) + 
    \frac{\log\!\left(\!\Det\left\{\!\Rh\R\!\right\}\!\right)}{2}=
    \\
    &\log\!\!\left(\!\!\Det\left\{\!\!\frac{\Rh\!\!+\!\R}{2}\!\!\right\}
    \!\!\Det\left\{\!\Rh{}^{-1}\R{}^{-1}\!\right\}
    \!\!\!\right) \!+ \!
    \frac{\log\!\left(\!\!\Det\left\{\!\Rh\R\!\right\}\!\!\right)}{2}=
    \\
    &\log\!\!\left(\!\!\Det\left\{\!\!\frac{\Rh\!\!+\!\R}{2}\!\!\right\}
    /\Det\left\{\!\Rh\R\!\right\}
    \!\!\right) \!+ \!
    \frac{\log\!\left(\!\!\Det\left\{\!\Rh\R\!\right\}\!\!\right)}{2}=
    \\
    &\log\!\!\left(\!\!\Det\left\{\!\!\frac{\Rh\!\!+\!\R}{2}\!\!\right\}
    \!\!\right) \!\!-\! \!
    \log\!\left(\!\!\Det\left\{\!\Rh\R\!\right\}\!\!\right)
    \!\!+\!\!
    \frac{\log\!\left(\!\!\Det\left\{\!\Rh\R\!\right\}\!\!\right)}{2}
    \\
    &=d_\text{LD}^2\left( \Rh,\R \right).
\end{align}
This results proves that the Log-Determinant metric is invariant to inversion. 

\end{document}